\documentclass[
    twocolumn,
    english,
    aps,
    pra,
    longbibliography,
    superscriptaddress,
    amsmath,
    amssymb,
    floatfix,
    nofootinbib,
    10pt
]{revtex4-2}

\usepackage[colorlinks=true,citecolor=teal,linkcolor=purple,urlcolor=brown]{hyperref}

\usepackage[utf8]{inputenc}
\usepackage{
    algorithm,
    algpseudocode,
    amsmath,
    amssymb,
    amsthm,
    csquotes,
    dsfont,
    graphicx,
    mathtools,
    microtype,
    natbib,
    nicefrac,
    physics,
    url,times,
    xcolor
}

\usepackage{tikz}
\usetikzlibrary{quantikz2}

\usepackage{hyperref}

\newtheorem{theorem}{Theorem}

\newtheorem{lemma}[theorem]{Lemma}

\newtheorem{corollary}[theorem]{Corollary}

\newtheorem{innercustomgeneric}{\customgenericname}
\providecommand{\customgenericname}{}
\newcommand{\newcustomtheorem}[2]{
  \newenvironment{#1}[1]
  {
   \renewcommand\customgenericname{#2}
   \renewcommand\theinnercustomgeneric{##1}
   \innercustomgeneric
  }
  {\endinnercustomgeneric}
}

\newcustomtheorem{customthm}{Theorem}
\newcustomtheorem{customlemma}{Lemma}
\newcustomtheorem{customproposition}{Proposition}
\newcustomtheorem{customcorollary}{Corollary}
\newcustomtheorem{customdefinition}{Definition}

\newenvironment{remark}[1][Remark]{\begin{trivlist}
\item[\hskip \labelsep {\bfseries #1}]}{\end{trivlist}}

\newcommand{\DQI}{\ensuremath{\lvert\mathrm{DQI}\rangle}}
\newcommand{\PDQI}{\ensuremath{\calP_\mathrm{DQI}}}
\newcommand{\sDQI}{\ensuremath{\langle s\rangle_\mathrm{DQI}}}

\DeclareMathOperator{\avg}{avg}

\providecommand{\tc}{\ensuremath{\Tilde{c}}}

\providecommand{\tf}{\ensuremath{\Tilde{f}}}

\providecommand{\to}{\ensuremath{\Tilde{o}}}

\providecommand{\calP}{\ensuremath{\mathcal{P}}}

\providecommand{\calS}{\ensuremath{\mathcal{S}}}

\DeclareMathOperator*{\bbE}{\ensuremath{\mathbb{E}}}
\providecommand{\bbF}{\ensuremath{\mathbb{F}}}

\providecommand{\bbN}{\ensuremath{\mathbb{N}}}

\DeclareMathOperator*{\bbP}{\ensuremath{\mathbb{P}}}

\providecommand{\bbZ}{\ensuremath{\mathbb{Z}}}

\newcommand{\F}{\mathbb{F}}

\newif\ifverbose
\verbosetrue

\newcommand{\stkout}[1]{\ifmmode\text{\st{\ensuremath{#1}}}\else\st{#1}\fi}

\graphicspath{figures/}

\newcommand{\fu}{Dahlem Center for Complex Quantum Systems, Freie Universit\"{a}t Berlin, 14195 Berlin, Germany}
\newcommand{\hzb}{Helmholtz-Zentrum Berlin f{\"u}r Materialien und Energie, 14109 Berlin, Germany}
\newcommand{\hhi}{Fraunhofer Heinrich Hertz Institute, 10587 Berlin, Germany}
\newcommand{\ibmzrl}{IBM Research, 8803 Rueschlikon, Switzerland}
\newcommand{\ibmisrael}{IBM Research, Haifa, 3498825, Israel}

\begin{document}

\title{Approximate sampling from decoded quantum interferometry\\via Markov chain Monte Carlo methods}
\date{\today}

\author{Elies Gil-Fuster}
\affiliation{\fu}
\affiliation{\hhi}
\affiliation{\ibmzrl}

\author{Matan Ninio}
\affiliation{\ibmisrael}

\author{Lennart Bittel}
\affiliation{\fu}

\author{Yishai Shimoni}
\affiliation{\ibmisrael}

\author{Jens Eisert}
\affiliation{\fu}
\affiliation{\hhi}
\affiliation{\hzb}

\author{Stefan Woerner}
\affiliation{\ibmzrl}

\author{Almudena Carrera V\'azquez}
\affiliation{\ibmzrl}

\begin{abstract}
Optimization problems are among the leading candidates for industrially relevant quantum advantage.
Decoded quantum interferometry (DQI) has been proposed to tackle approximate optimization, establishing a connection to classical decoding problems.
While previous work has primarily focused on the theoretical complexity of DQI, comparatively little is known about its empirical performance relative to classical algorithms.
In this work, we shed further light on the complexity of DQI and investigate numerically whether classical sampling methods can emulate the optimization capabilities of DQI.
We first present a simplified analytical characterization of DQI that connects its expected performance to binomial statistics, and we identify concrete obstacles in further studying the complexity of DQI.
Exploiting the fact that DQI output probabilities are efficiently computable, we apply
Markov chain Monte Carlo (MCMC) techniques, particularly block-Gibbs sampling, to sample from the induced distribution.
We study the runtime scaling of these methods for two optimization problems called max-XORSAT, where we reach beyond $1000$ effective qubits; and OPI, where we reach beyond $150$ effective qubits.
Our results show that MCMC algorithms can reliably attain the approximation ratios expected from DQI across a broad range of problem sizes.
In OPI, in the regime where a super-polynomial advantage is claimed for DQI, we observe an empirical runtime for MCMC that scales approximately as $1.1^{n}$, indicating exponential growth with a comparatively small base.
Our findings do not refute existing quantum advantage claims but provide new empirical evidence that classical sampling algorithms can closely match DQI's optimization performance, offering a more nuanced perspective on the practical advantage of DQI.
\end{abstract}

\maketitle

\section{Introduction
\label{s:introduction}}

\begin{figure}[t]
    \centering
    \includegraphics{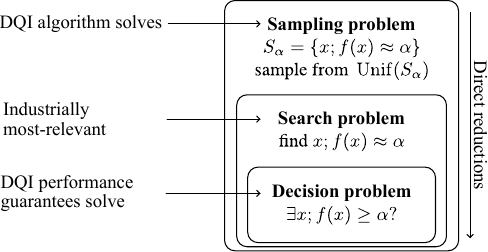}
    \caption{
        \textbf{Nested problems in combinatorial optimization.}
        Sketch for three possible natural problems arising in combinatorial optimization.
        In the context of \emph{decoded quantum interferometry} (DQI), both the \emph{decision} and \emph{sampling} problems can be solved efficiently for carefully chosen threshold $\alpha$ (see Eq.~\eqref{eq:dqi_score}).
        We tackle the \emph{search} problem with classical algorithms and report on their empirical efficiency.
    }
    \label{fig:summary_problems}
\end{figure}

Combinatorial optimization problems play a central role in numerous industrial applications, including routing, scheduling, and resource allocation.
As such, several attempts have been made in the form of concrete quantum algorithms \cite{MontanaroOverview,dalzell2023quantum,MindTheGaps} for combinatorial optimization~\cite{Abbas2024,OptimizationAdvantages,Jordan2025optimization,koch2025qoblib}.
One recent framework, dubbed \emph{decoded quantum interferometry} (DQI)~\cite{Jordan2025optimization}, has attracted much attention owing to both strong performance guarantees and a remarkably distinct flavor when compared to other quantum approaches to optimization.
Rather than directly outputting a single solution, DQI induces a probability distribution biased toward good solutions.
Furthermore, producing this distribution follows a connection to the decoding problem in error correction, re-purposing Regev's reduction~\cite{Regev2009}.

The DQI framework has been studied from different perspectives.
Several works characterized the possible limitations of the algorithm, both from complexity-theoretic considerations~\cite{marwaha2025complexitydecodedquantuminterferometry,kramer2026tight,kramer2026approximability,sun2026worst,anschuetz2025decoded,parekh2025maxcut_dqi}, robustness against hardware noise~\cite{bu2026decoded}, and specific industrial optimization problems~\cite{sabater2026towards,thelen2026constraint}.
In parallel, concrete improvements have been made in the design and implementation of the quantum algorithm~\cite{khattar2025dqi_opi,patamawisut2025quantum,chailloux2024softdecoders,chailloux2025opi,blanvillain2026quantum,horinaga2026worst}.
The framework has also been extended to other classes of classical optimization problems~\cite{gu2025algebraicgeometrycodesdecoded, bu2026multivariate} and other genuinely quantum computational tasks~\cite{schmidhuber2025hamiltonian_dqi,bu2026hamiltonian}.

For an optimization problem called \emph{optimal polynomial intersection} (OPI), Ref.~\cite{Jordan2025optimization} claimed a super-polynomial quantum speed-up for DQI against all known classical algorithms, in specific parameter regimes.
This claim has so far stood unchallenged.
At the same time, a distinct property of DQI is that it remains challenging to identify the individual quantum ingredients that yield the purported advantage.
While certain progress has been made in the front of analytically characterizing the complexity of DQI~\cite{marwaha2025complexitydecodedquantuminterferometry}, no final verdicts have been reached.
On the empirical front, precious little is known about precisely \emph{how} DQI performs against classical algorithms at different scales.

We identify a hierarchy of tasks within combinatorial optimization: decision, search, and sampling problems, which we illustrate in Fig.~\ref{fig:summary_problems}.
In broad terms, usual complexity-theoretic arguments apply to the decision problem, the search problem may be the more practically-relevant one, and the DQI algorithm natively solves the sampling problem.
The division between decision and search highlights that certain results in complexity theory may be of limited interest for industrially-relevant problems.
Further, the division between search and sampling is reminiscent of foundational results from early demonstrations of random circuit sampling~\cite{SupremacyReview} as also discussed in Ref.~\cite{marwaha2025complexitydecodedquantuminterferometry}.

In this work, we shed further light on the complexity of DQI, both analytically and numerically.
We first derive a novel simplified analysis of the performance of DQI, which we use to identify an obstacle to proving that sampling from the DQI distribution may be classically hard.
We next outline another obstacle showing that the search problem may not be directly reducible to the decision problem, as is sometimes the case.
Furthermore, we propose using classical \emph{sampling} algorithms to tackle \emph{optimization} problems, emulating the inner workings of the DQI algorithm.
In particular, we employ usual forms of \emph{Markov chain Monte Carlo} (MCMC) methods~\cite{liu2008,gilks1995markov,robert_monte_2004}. 
To this aim, we numerically test the performance of MCMC methods across sizes, reaching over $100$ qubits in the regime where the claims for super-polynomial advantage hold for DQI.

Our analytical results elucidate a deep fact: a necessary condition on the optimization problem for DQI to be successful -- that the dual error-correcting code has high distance -- immediately implies that the decision version of that problem is classically easy.
A direct consequence of this fact is that studying the complexity of DQI requires studying the search version of the optimization problem, which is unusual.
From our analysis, we also identify concrete obstacles on the usual path to showing that sampling is hard (via Stockmeyer's algorithm and reductions to hardness of approximate counting~\cite{SupremacyReview, aaronson_computational_2010}).
We conclude that different, potentially novel, tools are required to study the classical hardness of sampling from the output distribution of DQI.
In parallel, our numerical results offer first concrete glimpses into the empirical complexity of the optimization problems tackled by DQI.
We find that usual MCMC methods reach comparable performance to DQI in time $\propto1.1^n$, where $n$ is the number of qubits of the output state of DQI.
We thus do not challenge the claims of 
a super-polynomial advantage over known classical algorithms by DQI.
Rather, considering the small base of the exponential, our results call for a nuanced interpretation of the quantum advantage claims by DQI, as they may only become relevant for very large instances.

\section{Preliminaries}\label{s:preliminaries}

\subsection{Max-LINSAT and classical linear codes}\label{ss:max-LINSAT}

Suppose we are given a prime number $p$ and a corresponding field $\F_p$, a matrix $B \in \F_p^{m \times n}$, sets $F_i \subset \F_p$ of size $r_i = |F_i|$ for $i=1, \ldots, m$, and $F = \{F_1, \ldots, F_m\}$.
We refer to $n$ as the number of variables, $m$ as the number of constraints, and we specialize to the case where all sets are of the same size $r=\lvert F_i\rvert$, for $i\in\{1,\ldots,m\}$.
Then, max-LINSAT is defined as the problem of finding $x \in \F_p^n$ that satisfies as many of the $m$ constraints $\sum_{j=1}^n B_{i,j} x_j \in F_i$ as possible.
Let $f_i: \F_p \rightarrow \{-1, +1\}$ denote the indicator functions taking the value $+1$ if $\sum_{j=1}^n B_{i,j}x_j \in F_i$ and $-1$ otherwise.
The total score $f(x) = \sum_{i=1}^m f_i(\langle b_i,x\rangle)$ equals the number of satisfied constraints minus the number of unsatisfied ones, where $b_i$ denotes the $i^\text{th}$ row of $B$.
Several different problems can be specified for max-LINSAT, all revolving around $(\arg)\max_{x \in \F_p^n} f(x)$.
Below, we alternate between the objective function $f(x)$ and the number of satisfied constraints $s(x)$, which are related by $s(x) = (f(x)+m)/2$.
The best achievable fraction of satisfied constraints $s(x)/m$ is commonly referred to as the \emph{approximation ratio}.

The coefficient matrix $B\in\bbF_p^{m\times n}$ specifies an error-correcting code, which we refer to as the \emph{dual code} $C^\perp$.
The dual code $C^\perp$ is the linear code with $B^\intercal\in\bbF_p^{n\times m}$ as its parity-check matrix:
\begin{align}
    C^\perp&\coloneqq\{c\in\bbF_p^m\mid B^\intercal c=0\}.
\end{align}
The main relevant property in this work is the \emph{distance} of the code $d$, defined as the shortest Hamming distance between any two codewords $d\coloneqq \min_{c\neq c'\in C^\perp} \lVert c-c'\rVert_0$.
Given a noisy codeword $\tc=c+e$, arising from a valid codeword $c\in C^\perp$ and a non-trivial error $e\in\bbF_p^m$, the \emph{decoding problem} consists of identifying $e$.
For a given decoder, we say \enquote{it can decode $\ell$ errors} if the decoder is successful for all errors of Hamming weight up to $\ell$.
Foundational results in coding theory dictate that the maximum number of decodable errors fulfills $2\ell+1\leq d$.

Mirroring the analysis of Ref.~\cite{Jordan2025optimization}, we focus on two special cases of max-LINSAT, and their corresponding linear codes, which we introduce next.

\paragraph{Max-XORSAT.}

Max-XORSAT is the special case of max-LINSAT with $p=2$ and $r=1$.
In max-XORSAT, each constraint has the form $\sum_{j=1}^n B_{i,j} x_j = v_i$ for $v_i \in \F_2$ and $x \in \F_2^n$.  
The family of instances we study follows the setting of the original DQI paper~\cite{Jordan2025optimization}: the coefficient matrix $B \in \F_2^{m \times n}$ is drawn from a random irregularly sparse
ensemble with $m \approx 2n$ constraints.
The precise distribution over matrices $B$ is a variant of Gallagher's ensemble with fine-tuned sparsity.

The dual codes corresponding to sparse max-XORSAT are well-known as \emph{Low-Density Parity Check} (LDPC) codes.
Usual decoder algorithms for LDPC codes are based on Belief-Propagation (BP).
Calculating the exact distance of an individual LDPC code is hard, and instead one relies on heuristic and numerical estimations.
As a consequence, BP-based decoders may fail on a small fraction $\varepsilon$ of errors of a given Hamming weight $\ell$.

\paragraph{Optimal polynomial intersection.}

Optimal polynomial intersection (OPI) is the special case of max-LINSAT where the coefficient-matrix $B\in\bbF_p^{m\times n}$ is a \emph{Vandermonde matrix}: given a primitive root $\gamma\in\bbF_p^\ast$, we have $B_{i,j}=\gamma^{i(j-1)}$.
Following the original reference~\cite{Jordan2025optimization}, we set the number of constraints to be $m = p-1$, and the number of variables $n$ as $n\approx p/2$.
The name of this problem comes from the following observation: given a primitive root $\gamma\in\bbF_p^\ast$, any symbol string $x\in\bbF_p^n$ can be interpreted as the $n$ coefficients of a polynomial $Q$ of degree $n-1$.
The constraint for each $i\in\{1,\ldots,m\}$ then asks that the image $Q(\gamma^i) = \sum_{j=1}^n x_j \gamma^{i(j-1)}$ lies in the set $F_i$.

A linear code $C^\perp$ whose parity-check matrix $B^\intercal\in\bbF_p^{n\times m}$ is a Vandermonde matrix is known as a Reed-Solomon code with distance $n+1$.
In terms of decoders, the Berlekamp-Massey algorithm~\cite{Berlekamp2015} decodes up to $\ell=\lfloor n/2\rfloor$ errors deterministically.
We summarize the instance parameters in Table~\ref{tab:OPI_instance}.

\begin{table}[h]
    \centering
    \caption{\textbf{Summary of instance parameters for OPI.}}
    \begin{tabular}{c|c}\hline\hline
        \bf Notation & \bf Explanation \\\hline\hline
        $p$ & Prime, specifies the field $\bbF_p$ \\
        $\gamma\in\bbF_p^\ast$ & Primitive root, specifies the constraints $Q(\gamma^i)\in F_i$ \\
        $m=p-1$ & Domain size of polynomial $Q$, number of constraints \\
        $n\approx p/2$ & degree of $Q$ $+1$, number of variables \\
        $B_{i,j}=\gamma^{i(j-1)}$ & Vandermonde coefficient matrix \\
        $\ell=\lfloor n/2\rfloor$ & Number of decodable errors for Reed-Solomon code \\\hline\hline
    \end{tabular}
    \label{tab:OPI_instance}
\end{table}

\subsection{Decoded quantum interferometry \label{ss:DQI}}

Decoded quantum interferometry (DQI) is a quantum algorithm that produces good solutions to certain families of max-LINSAT instances~\cite{Jordan2025optimization}.
The inner workings of the algorithm instantiate Regev's reduction~\cite{Regev2009}, reducing the optimization problem to the decoding problem in error correction via the quantum Fourier transform.
The main goal of DQI is to prepare a quantum state which, upon measurement in the computational basis, induces a distribution over symbol strings $x\in\bbF_p^n$ that is biased toward higher scores $f(x)$.

More precisely, consider a degree-$\ell$ polynomial of the score function $P(f)=\sum_{k=0}^\ell \alpha_k f^k$, where $\alpha_0,\ldots,\alpha_\ell$ are the real-valued coefficients.
Then, the amplitudes of the output state of DQI, which we refer to as $\DQI$, are captured by such a polynomial:
\begin{align}
    \DQI \coloneqq \sum_{x\in\bbF_p^n} P(f(x))\lvert x\rangle = \sum_{x\in\bbF_p^n}\sum_{k=0}^\ell \alpha_k f^k(x)\lvert x\rangle,
\end{align}
where the coefficients $(\alpha_k)_{k=0}^\ell$ are chosen to ensure normalization.
Here, symbol strings $x \in \F_p^n$ are encoded using $n \lceil\log_2(p)\rceil$ qubits.
For DQI to be successful, the main requirement is to provide a decoder capable of decoding $\ell$-weight errors coherently in superposition~\cite{Jordan2025optimization}.
Importantly: the degree of the polynomial $\ell$ is exactly the number of decodable errors.
The probability of observing $x\in\bbF_p^n$ from a computational basis measurement is $\calP(x)=P^2(f(x))$.

We next discuss the performance guarantees of DQI for max-LINSAT, as presented in Ref.~\cite{Jordan2025optimization}.
We denote by $\PDQI$ the distribution over symbol strings induced by DQI, and by $\sDQI=\bbE_{x\sim\PDQI}[s(x)]$ the \emph{expected} number of satisfied constraints.
With these, the main figure of merit is the expected \emph{fraction} of satisfied constraints $\sDQI/m$, which dictates the approximation ratio achieved by DQI.

Suppose a classical polynomial-time decoder that can correct up to $\ell$ symbol errors in codewords from the code $C^{\perp} = \{ c \in \bbF_p^m : B^\intercal c = 0\}$.
Then, in the limit $m,\ell\to\infty$ with fixed ratio $\ell/m\in(0,1)$, the expected fraction of satisfied constraints from measuring the output state of DQI is
\begin{align}\label{eq:dqi_score}
    \frac{\sDQI}{m} &=
    \left( \sqrt{\frac{\ell}{m}\left(1 - \frac{r}{p}\right)} + \sqrt{\frac{r}{p}\left(1 - \frac{\ell}{m}\right)} \right)^2,
\end{align}
if $r/p\leq 1-\ell/m$, and $\sDQI/m=1$ otherwise.
Here the polynomial $P$ of degree $\ell$ is chosen specifically to maximize this expectation value (see~Appendix~\ref{a:optimal_polynomial} for its explicit construction).
Notably, both $\sDQI$ and the optimal polynomial $P$ are independent of the choice of $F$, as we discuss in Appendix~\ref{a:simplified_analysis}.

This formula allows us to predict the expected performance of DQI for problems where we have exact decoders.
For the instance parameters we select for OPI (see Table~\ref{tab:OPI_instance}), the expected fraction of satisfied constraints is $>0.9$.
When only an imperfect decoder is available (one that fails on a fraction $\varepsilon$ of the errors) the performance guarantee in the case $p=2,r=1$ becomes
\begin{align}\label{eq:dqi_score_eps}
    \frac{\langle s\rangle_{\avg}}{m} &\geq \frac{1}{2}+\sqrt{\frac{\ell}{m}\left(1-\frac{\ell}{m}\right)} - \varepsilon.
\end{align}
This formula also holds in the limit of $\ell,m\to\infty$, with a fixed ratio $\ell/m\in(0,1)$, and the average is taken over uniformly random sets $F$.
For random sparse max-XORSAT, we can obtain the expected fraction of satisfied constraints by numerically estimating the fraction $\varepsilon$ of decodable errors for different weights $\ell$.
We use the notation $\langle s\rangle_{\avg}$ to distinguish this average-case guarantee from the worst-case guarantee in Eq.~\eqref{eq:dqi_score}, which holds for all choices of $F$.
The expression in Eq.~\eqref{eq:dqi_score_eps} presents a trade-off: increasing $\ell$ raises the term inside the square root, but it can also increase the decoder failure probability $\varepsilon$.
Table~\ref{tab:summary} summarizes the two problem families and their DQI performance guarantees.

\begin{table*}
    \centering
    \caption{
        \textbf{Special cases of max-LINSAT addressed in this work}.
        We relate the optimization problems to their dual error-correcting codes, as well as the proven performance guarantees using the 
        DQI algorithm (matching Ref.~\cite{Jordan2025optimization}).
    }\label{tab:summary}
    \begin{tabular}{c|c|c|c|c|c}\hline\hline
        Optimization & Error-correcting & Decoder & Fraction of correctly & Expected & Statement \\
        problem & code & & decoded errors & DQI score & type \\ \hline\hline
        OPI & Reed-Solomon & Berlekamp-Massey & $1$ & Eq.~\eqref{eq:dqi_score} & worst-case over $F$ \\\hline
        sparse max-XORSAT & LDPC & Belief propagation & $1-\varepsilon$ & Eq.~\eqref{eq:dqi_score_eps} & average-case over $F$ \\ \hline\hline
    \end{tabular}
\end{table*}

\subsection{Decision, search, and sampling}\label{ss:dec_search_samp}

It is worth noting that the majority of results in complexity theory concern \emph{decision} problems.
As we explain below, studying the complexity of DQI may require considering \emph{search} or \emph{sampling} problems instead; we briefly introduce these below.

Let us consider an abstract optimization problem specified by an objective function $f:\bbF_p^n\to\bbZ$.
One natural way to define a \emph{decision} problem would be, given a threshold $\alpha\in\bbZ$, decide between:
\begin{itemize}
    \item \textbf{YES}: if there exists $x\in\bbF_p^n$ such that $f(x)\geq\alpha$.
    \item \textbf{NO}: for all $x\in\bbF_p^n$, it holds $f(x) < \alpha$.
\end{itemize}
One natural way to define a \emph{search} problem would be, given a threshold $\alpha\in\bbZ$, find $x\in\bbF_p^n$ such that $f(x)\approx\alpha$ (assuming that such a symbol string exists).
Finally, one natural way to define a \emph{sampling} problem would be, given a threshold $\alpha\in\bbZ$ defining a subset of \enquote{good} symbol strings
\begin{align}
    S_\alpha=\{x\in\bbF_p^n\mid f(x)\approx\alpha\},
\end{align}
sample from the uniform distribution over $S_\alpha$.
As sketched in Fig.~\ref{fig:summary_problems}, these problems admit straightforward reductions: solving the sampling task immediately solves the search task, and solving the search task immediately solves the decision task.
We note that, for certain families of optimization tasks, the reverse directions may also hold.

\subsection{Related work on the complexity of max-LINSAT \& DQI}\label{ss:related}

Since max-LINSAT contains famously-hard optimization problems as special cases~\cite{hastad2001optimal}, we do not expect polynomial-time (classical or quantum) algorithms to perform better than random guessing in the worst case.
The main question is, rather: for which classes of problems does DQI perform better than known classical algorithms.
For those problems, we also ask whether the performance of DQI provably surpasses all possible polynomial-time classical algorithms.

The results in the foundational work~\cite{Jordan2025optimization} touch on all three versions of the max-LINSAT optimization problem.
By construction, DQI addresses the \emph{sampling} version of max-LINSAT, since the algorithm prepares a quantum state where different symbol strings achieving the same score have the same amplitude.
As we discussed, DQI then also addresses the \emph{search} version of max-LINSAT, by producing a single sample, with high probability.
Further, the performance guarantees of DQI~\cite{Jordan2025optimization} address the \emph{decision} version of max-LINSAT, for a carefully chosen threshold $\alpha$ depending on the parameters $\ell, m, r, p$ (see Eq.~\eqref{eq:dqi_score}).
The proofs involve the basis of symmetric polynomials.

The authors of Ref.~\cite{marwaha2025complexitydecodedquantuminterferometry} discuss further the complexity of \emph{sampling}, and they place approximate DQI sampling in $\mathsf{BPP^{NP}}$.
Their arguments hint that the DQI distribution is well-approximated by the uniform distribution over an \enquote{unstructured} subset $S\subseteq\bbF_p^n$ of size $\lvert S\rvert=p^{cn}$, with $c\in(0,1)$.
This means that $S$ is both exponentially large \emph{and} only an exponentially vanishing fraction of all symbol strings.
For this argument, the authors prove that the output state of DQI is not peaked, for which they use the basis of \emph{Kravchuk} polynomials, which generalizes symmetric polynomials.
Interestingly, the mechanisms that could make sampling hard in this case are rather different from previous cases, for instance those presented in Ref.~\cite{SupremacyReview} around the earlier claims for quantum advantage in sampling from random quantum circuits or schemes like boson sampling.
In particular, showing rigorously that quantum random sampling is hard (under mild assumptions) involves a reduction to a counting problem: it is widely believed that computing individual output probabilities for random quantum circuits is $\#\mathsf{P}$-hard.

In parallel, Refs.~\cite{chailloux2024softdecoders,chailloux2025opi,blanvillain2026quantum} proposed improving DQI by using quantum decoders, as opposed to coherent classical decoders.
These works then show a higher expected performance on average over instances, by reaching further than the threshold of perfect decoding.
Further, recent works have studied the complexity of restricted families of instances of max-LINSAT.
On the one hand, Refs.~\cite{kramer2026approximability,kramer2026tight} characterize the worst-case hardness of max-LINSAT both as a function of the ratio $r/p$ and the sparsity of $B$.
These results identify regimes where no polynomial-time algorithm is expected to perform better than random guessing, and thus rule out worst-case super-polynomial speed-ups using DQI.
On the other hand, Refs.~\cite{sun2026worst,horinaga2026worst} study instances of OPI where the worst-case performance of DQI is provably sub-optimal, namely where there exist solutions with scores beyond the reach of DQI.
First, Ref.~\cite{sun2026worst} proved the existence of such instances, then and recently Refs.~\cite{horinaga2026worst, jo2026efficient} improved on these results.
Specifically, the authors of Ref.~\cite{horinaga2026worst} construct a polynomial-time quantum algorithm for these beyond-DQI instances of OPI, and exact sampling quantum algorithms are provided in Ref.~\cite{jo2026efficient}.

\section{Markov chain Monte Carlo for max-LINSAT
\label{s:MCMC}}

\emph{Markov-chain Monte-Carlo} (MCMC) methods are a family of algorithms for sampling from complex probability distributions $\mu$ over $x \in \F_p^n$ when direct sampling is not possible, typically due to high dimensionality or because $\mu$ is only specified up to an unknown normalization constant \cite{liu2008,gilks1995markov,robert_monte_2004}. 
MCMC constructs a Markov chain whose stationary distribution is the target $\mu$, so that after a sufficiently long \emph{mixing time}, its samples approximate draws from $\mu$. 
A key advantage of these methods is that they require only local evaluations of unnormalized probabilities or energy differences, making them both flexible and widely applicable.
Among the many variants of MCMC, we focus here exclusively on Gibbs sampling.

\subsection{Gibbs sampling}\label{s:gibbs_sampling}

Gibbs sampling leverages conditional distributions to approximate the target distribution. 
A key advantage is that it avoids rejection steps, provided that sampling from the conditionals is computationally feasible.
The method iteratively selects a coordinate $j\in\{1,\ldots,n\}$ and resamples $x_j$ from the conditional distribution
\begin{align}
    \bbP[x_j = q \mid x_{/j}] &=
    \frac{\mu(x_j^q)}{\sum_{q'=0}^{p-1} \mu(x_j^{q'})},
\end{align}
where $x_{/j}$ denotes all components of $x\in \F_p^n$ except the $j$-th, and $x_j^q$ is $x$ with the $j$-th component replaced by $q\in \F_p$.
For the DQI target distribution $\mu(x) = \PDQI(x)$ with $\PDQI(x)=P^2(f(x))$ (as introduced in Section~\ref{ss:DQI}), this requires evaluating the polynomial $P(f(x))$ for $p$ different candidate values of $x_j$.
For efficient evaluation of the polynomial $P(f(x))$ as a function of $x$ in the MCMC context, see Appendix~\ref{a:efficient_probabilities}.

A natural generalization referred to as \emph{block Gibbs sampling} is to update a block of $\kappa$ variables simultaneously, sampling from the joint conditional distribution over the block.
This requires evaluating $\mathcal{O}(p^\kappa)$ conditional probabilities but can substantially improve convergence by moving the chain longer distances in a single step.
In all numerical experiments presented in this paper, we use a block size of $\kappa=3$, which provides a good balance between computational cost and convergence speed.
Gibbs sampling (and its block variant) produces a Markov chain that, under standard conditions such as irreducibility and aperiodicity, converges to the target distribution~\cite{robert_monte_2004}.

\subsection{Approximate simulation of the DQI state}\label{ss:simulation}

    We consider the possibility of simulating DQI with classical MCMC methods.
    To be precise, we do not try to simulate each step of the quantum circuit.
    Rather, we aim to reproduce the distribution induced by the DQI state.
    As discussed in Refs.~\cite{Jordan2025optimization, marwaha2025complexitydecodedquantuminterferometry}, the amplitudes of the DQI state are efficiently classically computable, and this allows us to use off-the-shelf MCMC methods.

    In studies of classical simulation of quantum computation, usual figures of merit involve measures of infidelity or distance in the space of distributions.
    These are for instance the metrics used in Ref.~\cite{marwaha2025complexitydecodedquantuminterferometry} to qualify the classical hardness of sampling from the DQI distribution.
    Instead, we turn our attention to the relevant figure of merit for the optimization task: namely, achieving a score comparable to that of DQI.

    This way, we flexibly navigate the boundary between simulation and optimization: we use an approximate \emph{simulation} algorithm and assess its performance according to the \emph{optimization} problem.
    We would deem an approximate simulation algorithm successful if it matched the performance of DQI in polynomial time (provably or heuristically).
    If such an approach were successful, it would challenge the claims of a super-polynomial speed-up by DQI~\cite{Jordan2025optimization}.
    Conversely, this hypothetical scenario would not directly contradict the complexity-theoretic results of Refs.~\cite{marwaha2025complexitydecodedquantuminterferometry} (a classical sampling algorithm could match the performance of DQI without being a good approximation for the entire distribution) nor~\cite{kramer2026approximability,kramer2026tight,sun2026worst,horinaga2026worst} (both DQI and the purported simulator would run in polynomial time).

    \begin{remark}[Equivalent number of qubits.]
        In max-XORSAT, the number of variables $n$ corresponds to the number of qubits that comprise the DQI state.
        In OPI, this is no longer the case.
        Rather, every element of $\bbF_p$ can be represented with $\lceil \log_2p\rceil$ qubits.
        Hence, for the problem defined over $\bbF_p^n$, the DQI state uses $n_p\coloneqq n\lceil \log_2p\rceil$ qubits.
        In this work we focus on the (equivalent) \emph{number of qubits} for both max-XORSAT and OPI, as this is the relevant scaling parameter for the runtime of the quantum algorithm~\cite{khattar2025dqi_opi}.
    \end{remark}

\subsection{Search versus sampling with Markov chain Monte Carlo}\label{ss:mcmc-optim-samp}

    Tackling the \emph{search} problem with MCMC methods follows a straightforward procedure: fix a score threshold $\alpha$, and set the following stopping criterion for the Markov chain: \enquote{$f(x_\tau)\geq\alpha$}.
    We refer to $\tau$ as the \emph{mixing time}: the number of steps required to reach the stopping criterion.
    Also, for ease of language, we refer to a sample $x$ as \enquote{a good sample} if it satisfies $f(x)\geq\alpha$.

    Regarding the \emph{sampling} problem, we first fix a concrete goal for the algorithm.
    Recall from Ref.~\cite{marwaha2025complexitydecodedquantuminterferometry} that the DQI distribution is well-approximated by the uniform distribution over good samples.
    We call $S$ the set of good samples, with $\lvert S\rvert=p^{cn}$, for $c<1$.
    Then, approximating the uniform distribution over $S$ is sufficient to approximate the DQI distribution, to some extent.
    The standard approach to testing uniformity \emph{from samples} is based on the birthday paradox: the collision probability of two independent samples is
    \begin{align}
    \bbP_{x,x'\sim\mathcal{P}}\!\left[x=x'\right] &= \frac{1}{\lvert S\rvert} = p^{-cn},
    \end{align}
    so one typically requires \(\Theta(\sqrt{\lvert S\rvert}) = \Theta(p^{cn/2})\) samples before observing a collision.
    We note that avoiding collisions is only a one-sided test, but to make the discussion concrete, we conceptualize the goal of sampling as: \enquote{find $N$ \emph{different} good samples.}
    We call $\tau_N$ the \emph{sampling time}: the total number of steps required to produce $N$ different good samples.

\subsection{Sampling algorithms with Markov chain Monte Carlo}\label{ss:sampling_mcmc}

    We identify two different possible procedures for the sampling task, which we refer to as \enquote{restart} and \enquote{keep-going}.
    In \emph{restart}, we run an independent Markov chain for each good sample.
    Assuming sequential computation, this means that every time we find a new sample satisfying the score condition, we simply discard the Markov chain and start with a fresh run.
    In \emph{keep-going}, we only stop the Markov chain once $N$ different good samples have been found, and we record every new good sample along the way.
    We briefly discuss the potential relative merits of the two sampling approaches.
    
    Independent MCMC runs are expected to produce independent outputs, and given the symmetry among the good samples (all samples achieving the same score receive the same probability), this should indicate that the \emph{restart} algorithm indeed produces uniform-random samples.
    It could be that, due to a complex score landscape, MCMC can only reach certain subsets of good samples, but as long as these are large enough, we would not be able to distinguish them from complete uniform random in polynomial time.
    
    Next, we consider the landscape of scores for all possible symbol strings.
    We can interpret MCMC as a probabilistic hill-climbing algorithm, which starting from a random point on the map, performs random local moves that are biased toward higher scores.
    Let us assume that the approach eventually succeeds and we reach a good sample, and let us say further we are in a local maximum.
    From here, the \emph{keep-going} algorithm ought to find other \emph{different} good samples.
    In the case that the good samples cluster in sizable connected components, yielding a ridge-like hilltop, we could expect the algorithm to perform a random walk over the hilltop, and hence avoid repeatedly coming back to the same good sample.
    Whether this clustering happens in practice is related to the so-called \emph{overlap gap property}, as discussed for instance in Ref.~\cite{anschuetz2025decoded}.
    In the general case, where no such clustering is assumed, to keep-going may involve some hurdles.
    For the algorithm to produce a new good sample, it first needs to climb down from the hilltop, reach a valley, and then climb up in a different direction.
    Yet, the random local moves of MCMC are biased towards climbing \emph{up}, not \emph{down}: the probability that the chain abandons a hilltop in favor of a valley is probabilistically suppressed.
    In practice, different problem instances may come with varying necessary times for random samples to abandon their respective hills, and we study these numerically.

    Call $\tau$ the \emph{mixing time} of the chain: the number of steps necessary to reach one good sample starting from scratch.
    Also, call $I$ the \emph{intra-sample time}: the number of steps necessary to go from one good sample to a different one (as we do in the \emph{keep-going} procedure).
    Then, the expected total number of steps for \emph{restart} is $N\tau$: the number of steps for each of them is exactly the mixing time.
    The expected total number of steps for \emph{keep-going} is $\tau+(N-1)I$: we pay the cost of mixing once to find the first good sample, and then we hop between different good samples $N-1$ times, at a cost of $I$ each, for a total of $N$ different good samples.
    Relating to the clustering of solutions, different problems will come with different mixing and intra-sample times, from which we can argue which of the two sampling algorithms ought to be better: \emph{restart} should be favored when $\tau<I$, and \emph{keep-going} otherwise.

\section{Analytical Results}\label{s:analytical}

\subsection{Decision: moments of max-LINSAT \& DQI}\label{ss:moments}

We briefly present a simplified analysis on the expected performance of DQI.
In particular, a key feature in the original proofs in Ref.~\cite{Jordan2025optimization} was working on the basis of elementary symmetric polynomials, combined with several standard Fourier computations.
We provide a novel derivation which avoids symmetric polynomials, thus also complementing the approaches via Kravchuk polynomials in Refs.~\cite{marwaha2025complexitydecodedquantuminterferometry, sun2026worst,horinaga2026worst}.
We provide full details in Appendix~\ref{a:simplified_analysis}.

For ease of notation, we write $\PDQI$ as a polynomial of degree $2\ell$:
\begin{align}
    \PDQI(x) &\coloneqq P(f(x))^2 = \sum_{k'=0}^{2\ell}\beta_{k'} f^{k'}(x).
\end{align}
We use the fact that $P$ is a polynomial of degree $\ell$, and the coefficients $\beta_{k'}$ are the appropriate sums over products of the original coefficients $\alpha_k$.

\begin{lemma}[DQI vs uniform distribution -- informal]\label{l:moments_dqi_unif}
    Suppose we use DQI with a polynomial of degree $\ell$.
    Then, the $k^\text{th}$ moment of the objective function $f(x)$ under the DQI distribution is a linear combination of moments under the uniform distribution as
    \begin{align}
        \bbE_{x\sim \PDQI}\left[f^k(x)\right] &= p^n \sum_{k'=0}^{2\ell}\beta_{k'}\bbE_{x\sim\operatorname{Unif}}\left[f^{k+k'}(x)\right].
    \end{align}
\end{lemma}
\begin{proof}[Proof sketch]
    Full details in Appendix~\ref{a:simplified_analysis}.
    The main idea is that $\PDQI(x)f^k(x)$ is itself a polynomial of degree $2\ell+k$.
    Then, linearity allows us to decompose $\bbE_{x\sim\PDQI}[f^k(x)]$ into monomial terms $\bbE_{x\sim\operatorname{Unif}}[f^{k+k'}(x)]$, for each $k'\in\{0,\ldots,2\ell\}$.
\end{proof}

    \begin{theorem}[Moments of max-LINSAT -- informal]\label{thm:distance_uniform_distribution_moments_maxLINSAT}
        Suppose a max-LINSAT instance whose dual code has distance $d$.
        Let $J\sim\operatorname{Binom}(m, r/p)$ be a binomial-distributed random variable.
        Then, for all $k<d$, the $k^\text{th}$ moment of the objective function $f(x)$ under the uniform distribution fulfills
        \begin{align}
            \bbE_{x\sim\operatorname{Unif}}[f^k(x)] &= \bbE_{J}\left[(2J-m)^k\right].
        \end{align}
    \end{theorem}
    \begin{proof}[Proof sketch]
        Full details in Appendix~\ref{a:simplified_analysis}.
        We first expand $f^k(x)$ in terms of the individual summands that define the objective function $f(x)=\sum_{i=1}^m f_i(\langle b_i,x\rangle)$.
        We interpret the corresponding result as a sum over length-$k$ codewords from the dual code.
        Using the definition of distance and $k<d$, we simplify the sum to consider only the different ways in which the all-$0$ codeword can be obtained.
        We finally manipulate the resulting combinatorial sum as the expectation value of a binomial-distributed random variable.
    \end{proof}

    \begin{theorem}[Moments of the DQI distribution -- informal]\label{thm:DQI_performance_simplified}
        Suppose a max-LINSAT instance whose dual code has distance $d$.
        Suppose we use DQI with a polynomial of degree $\ell$, with $2\ell+1\leq d$.
        Let $J\sim\operatorname{Binom}(m,r/p)$ be a binomial-distributed random variable.
        Then, for any $k < d - 2\ell$, the $k^\text{th}$ moment  of the objective function $f$ under the DQI distribution fulfills
\begin{align}
            \bbE_{x\sim \PDQI}\left[f(x)^k\right] &= p^n  \bbE_{J}\left[\sum_{k'=0}^{2\ell}\beta_{k'}(2J -m)^{k+k'}\right].
        \end{align}
    \end{theorem}
    \begin{proof}[Proof sketch]
        Full details in Appendix~\ref{a:simplified_analysis}.
        The statement follows directly from Lemma~\ref{l:moments_dqi_unif} and Theorem~\ref{thm:distance_uniform_distribution_moments_maxLINSAT}.
    \end{proof}

    We note that Theorem~\ref{thm:DQI_performance_simplified} is a direct statement about the DQI distribution.
    This result draws its value from its formulation in the basis of monomials (as opposed to a representation in terms of symmetric polynomials).
    In addition, we believe that the proof is simpler than the previously existing ones from Refs.~\cite{Jordan2025optimization, sun2026worst,horinaga2026worst}.
    
    Conversely, we highlight Theorem~\ref{thm:distance_uniform_distribution_moments_maxLINSAT} as a statement about the complexity of max-LINSAT, independently of which optimizer we use.
    In particular, we identify the requirement \enquote{the dual code has distance $d$} as a concrete, quantifiable notion of \emph{structure} (aligned with results in Refs.~\cite{sun2026worst,horinaga2026worst}).
    Our result builds the intuition that the higher the distance, the closer the distribution of solutions is to being binomial-distributed, as more of their moments match.
    As discussed above, DQI requires the dual code to have high distance to perform better than random guessing, and hence the circle is closed: the performance of DQI is better than random guessing if the distribution of solutions is closer to a binomial.
    If the distribution of solutions is close to a binomial, though, then the decision version of the problem may become easier, especially for thresholds $\alpha$ within the bulk.
    One might have hoped that DQI would bring new insight into the study of $\mathsf{NP}$-hardness beyond certain approximation-ratio thresholds.
    Instead, our results indicate that the performance of DQI falls far below the \enquote{interesting} regime for complexity theory, namely the threshold for $\mathsf{NP}$-hardness.

\subsection{Obstructions for the complexity of DQI}\label{ss:obstructions}

Intuitively, technical subtleties arise when characterizing the complexity of DQI.
Precisely because DQI is a polynomial-time quantum algorithm that achieves $\sDQI$ for certain families of max-LINSAT, these problems cannot be arbitrarily hard, as standard belief in complexity theory maintains that $\mathsf{NP}\not\subseteq\mathsf{BQP}$.
This means that the complexity of the decision problem cannot resort to standard reductions to $\mathsf{NP}$-complete problems.
Further, the performance guarantees of DQI hold under the assumption that a polynomial-time classical decoder can decode a large fraction of errors up to weight $\ell$.
This is a property of the matrix $B$, and it introduces a restriction to the class of problems which should naturally be considered.

We next discuss the potential classical hardness of sampling from the DQI distribution in a way that is different from and strongly complementing the insights of Ref.~\cite{marwaha2025complexitydecodedquantuminterferometry}. 
Recall that an argument for the classical hardness of sampling from the DQI distribution was that the DQI distribution is well-approximated by the uniform distribution over an \emph{unstructured} subset of size $p^{cn}$, with $c<1$.
These ingredients, assuming the set $S$ were truly unstructured, would be enough to prove that sampling is hard.
Yet, $S$ is not unstructured in practice, as we have a succinct description for it: it is the set of all symbol strings $x$ whose number of satisfied constraints is close to $\sDQI$.
We consider the possibility of showing that sampling is hard without the assumption that $S$ is unstructured.
In Appendix~\ref{aa:counting} we propose a natural counting problem based on max-LINSAT, and we discuss a possible obstruction in showing that DQI can offer an advantage for approximate counting, which in turn would represent an obstruction in proving the classical hardness of sampling from $\PDQI$.
We highlight that the power of having access to samples from the DQI distribution can be matched by analytical arguments exploiting Theorem~\ref{thm:distance_uniform_distribution_moments_maxLINSAT}.

Regarding the search problem, we identify an obstruction against a standard reduction from search to decision (see Appendix~\ref{aa:decision} for further details).
To the best of our knowledge, this obstruction has not been spelled-out before.
Assume access to a machine capable of solving the decision problem for a given threshold $\alpha$.
A standard reduction from search to decision works iteratively, where at every step we query the decision machine while substituting additional entries of a growing \emph{guess} symbol string.
If the machine allows these modified queries for the same threshold $\alpha$, then a linear number of queries suffices to solve the search problem.
In the case of DQI, the performance guarantees (see Eq.~\eqref{eq:dqi_score}) could play the role of a (one-sided) decision machine.
We argue that this particular instantiation of a decision machine may be unable to address the necessary modified queries in the process we describe. 

\section{Numerical Results}\label{s:results}

\begin{figure*}[t]
    \centering
    \includegraphics{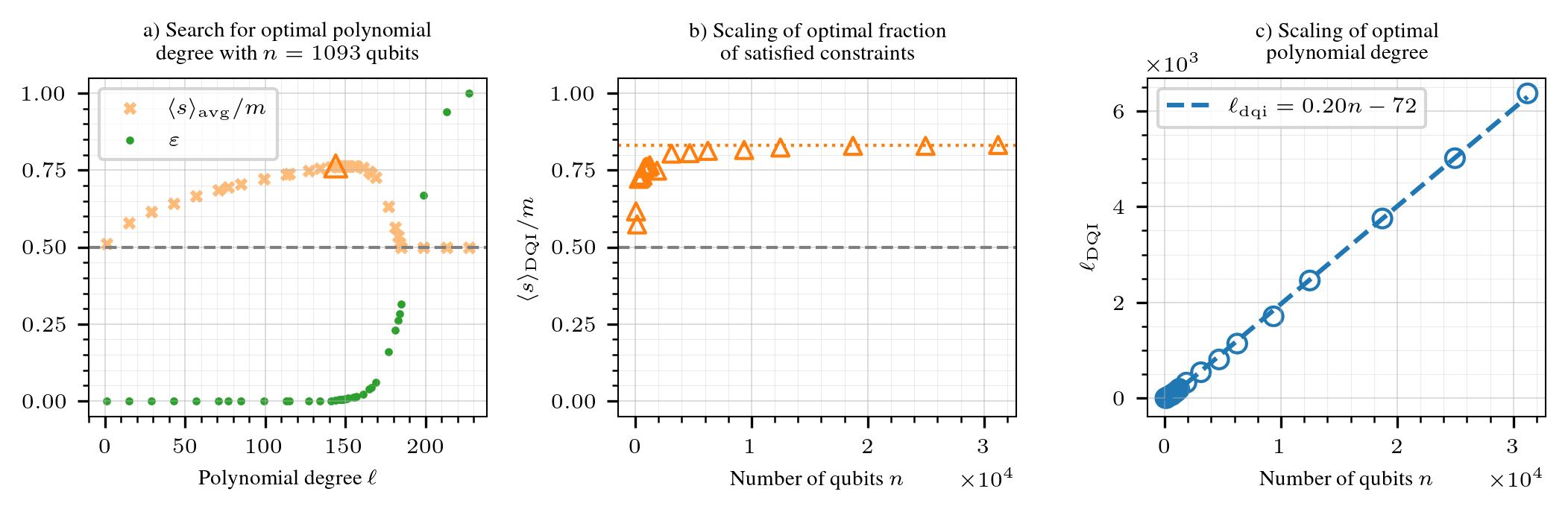}
    \caption{
        \textbf{Scaling behavior of DQI for max-XORSAT.}
        a) For each instance size, we search for the polynomial degree $\ell$ that maximizes the expected fraction of satisfied constraints $\langle s\rangle_\mathrm{avg}/m \geq 1/2 + \sqrt{\ell/m(1-\ell/m)}-\varepsilon$, where we estimate the fraction of incorrect decoding $\varepsilon$ numerically.
        The horizontal dashed line represents the performance floor that random guessing would achieve.
        The triangle highlights the optimal performance for this instance size.
        We observe that the fraction of incorrect decoding has a step-like behavior, which is consistent over different instance sizes (refer to Appendix~\ref{a:further_details} for further visualizations).
        b) As the instance size increases, the optimal fraction of satisfied constraints $\langle s\rangle_\mathrm{DQI}/m$ converges to $0.831$, as has been reported in Ref.~\cite{Jordan2025optimization}, and highlighted by the horizontal dotted line.
        c) The optimal polynomial degree $\ell_\mathrm{DQI}$ is well-approximated by a linear function of the number of qubits $n$ (or, equivalently, the number of constraints, $m\approx 2n$).
        These are the values we use for our numerical results with Markov chain Monte Carlo methods (see Section~\ref{s:results}), both for the polynomial degree $\ell_\mathrm{DQI}$ and performance threshold $\langle s\rangle_\mathrm{DQI}/m$.
    }
    \label{fig:ell_dqi_max-XORSAT}
\end{figure*}
    In this section, we present our comprehensive numerical results.
    We apply block Gibbs sampling with block size $\kappa=3$ to two special cases of max-LINSAT: max-XORSAT and \emph{optimal polynomial intersection} 
    (OPI). 
    We first specify the precise metrics and parameters for the experiment, and then discuss the runtime of our MCMC methods for the \emph{search} and \emph{sampling} versions of both \emph{max-XORSAT} and \emph{OPI}, as introduced in Sec.~\ref{s:preliminaries}.
    Recall that, to simplify the comparison and scaling statements, for OPI we report the equivalent number of qubits $n_p=n\lceil \log_2 p\rceil$.
    For our largest instances ($p=53,n=26$), this is $156$ qubits.

\subsection{Performance metrics}\label{ss:metrics}

    The following definitions follow the discussion in Sections~\ref{ss:DQI},~\ref{ss:dec_search_samp}, and~\ref{ss:mcmc-optim-samp}.
    Depending on the specifics of the problem, DQI succeeds either in the worst- or average-case accordingly to whether the decoder succeeds deterministically or probabilistically for all errors of weight up to $\ell$.
    Naturally, when pitting classical algorithms against the predicted DQI performance, the metrics used to judge the classical algorithms should reflect these facts.
    
\paragraph{Search.}
    For each problem family, we define a \emph{mixing time} $\tau$ as the number of Gibbs steps required for the chain to reach the DQI performance threshold $\langle s\rangle_{\mathrm{DQI}}$.
    Because the DQI guarantee for max-XORSAT holds on average over random choices of right-hand-side vector $v\in\mathbb{F}_2^m$, we define the \emph{average mixing time} $\tau_{\mathrm{avg}}$  as the average over $K=100$ random choices of $v$ of the first time the chain's score exceeds $\langle s\rangle_{\mathrm{DQI}}$.
    For OPI, where the guarantee holds in the worst case over all $F$, we define the \emph{maximum mixing time} $\tau_{\max}$ as the maximum over $K=100$ random choices of $F$.
    In both cases, we run $K$ chains independently and track the best score found up to each time step.
    For each chain $k$ we denote by $c_{k,t}$ the maximum score achieved up to step $t$.
    The mixing time is then the smallest $t$ such that the appropriate statistic (mean or min) of $\{c_{k,t}\}$ first exceeds $\langle s\rangle_{\mathrm{DQI}}$.
    
\paragraph{Sampling.}
    We define the \emph{sampling time} $\tau_N$ as the total number of Gibbs steps required to obtain $N=10$ different samples, all exceeding the DQI performance threshold $\sDQI$.
    For max-XORSAT, we report the average sampling time over $K=100$ random choices of $v$.
    For OPI, we average over $K=10$ random choices of $F$.

\subsection{Polynomial degree for max-XORSAT}\label{ss:poly_deg_max-XORSAT}

In Section~\ref{ss:DQI} we briefly mentioned the ensemble over which we draw the coefficient matrices for max-XORSAT.
Since exact deterministic formulas for the distance of these random codes are unavailable, we resort to numerical estimation.
For each instance size, we search for $\ell$ maximizing Eq.~\eqref{eq:dqi_score_eps}, thus striking a balance between the term in the square root and the fraction of incorrect decoding $\varepsilon$.
We find numerically that, for these instances, belief-propagation decoding achieves a decoder failure rate of $\varepsilon \approx 0$ for $\ell/m$ below a threshold of approximately $0.13$, at which point $\varepsilon$ jumps sharply (see the step-like behavior visible in Fig.~\ref{fig:ell_dqi_max-XORSAT}a, and Fig.~\ref{fig:ell_dqi_max-XORSAT_scaling} for further details).  
The optimal polynomial degree $\ell_{\mathrm{DQI}}$ lies just below this threshold, yielding $\langle s\rangle_{\mathrm{DQI}}/m \approx 0.831$ in the large-$m$ limit. 
For the numerical experiments, we determine $\ell_{\mathrm{DQI}}$ and $\langle s\rangle_{\mathrm{DQI}}/m$ empirically for each instance size by the procedure illustrated in Fig.~\ref{fig:ell_dqi_max-XORSAT}, and use these values as targets for the Gibbs chain.

\subsection{Mixing time and sampling time}

\begin{figure*}[t]
    \centering
    \includegraphics{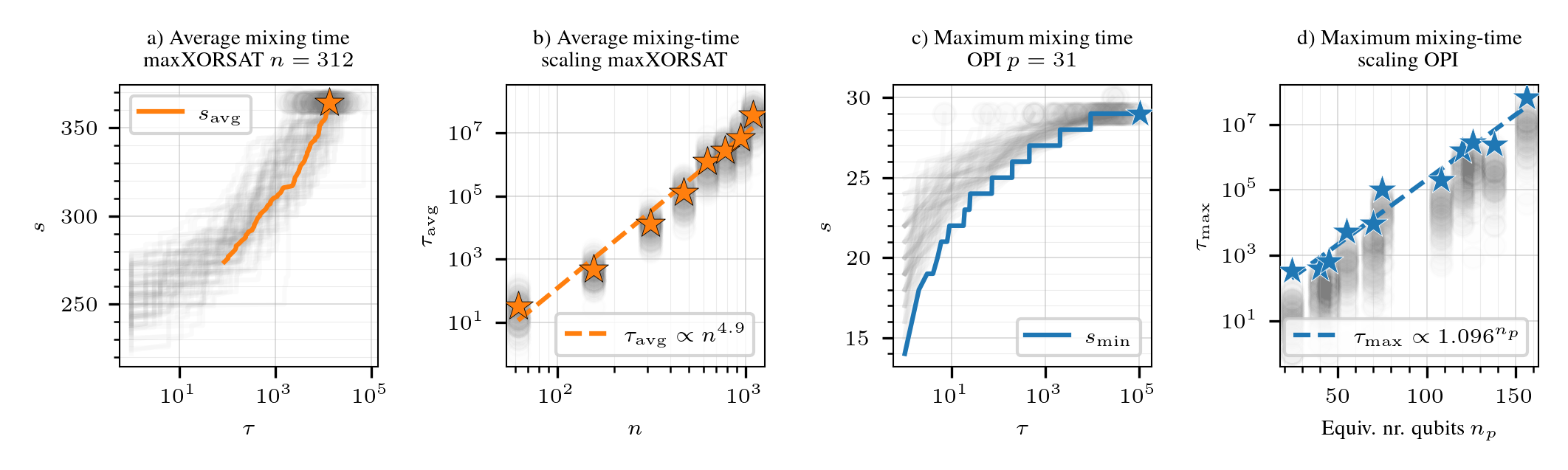}
    \caption{
        \textbf{Scaling behavior of Markov-chain mixing times.}
        a) The performance guarantees of DQI for max-XORSAT only hold on average over the choice of right-hand-side vector $v\in\bbF_2^m$.
        We sample $100$ uniform random vectors and monitor the average number of satisfied constraints $s_\mathrm{avg}$.
        We stop the chains once $s_\mathrm{avg}\geq\langle s\rangle_\mathrm{DQI}$, and we refer to the required number of steps as the average mixing time $\tau_\mathrm{avg}$ (see Fig.~\ref{fig:tau_search_scaling}a for further details).
        b) As the problem size increases, the average mixing time $\tau_\mathrm{avg}$ is well-approximated by a power law of degree $5$, which corresponds to a linear fit on a log-log plot.
        c) The performance guarantees of DQI for OPI hold independently over the choice of right-hand-side sets $F=\{F_i\}_{i=1}^m$, with $F_i\subseteq\bbF_p$, and $\lvert F_i\rvert=r\approx p/2$.
        We take $100$ uniform random sets and monitor the worst number of satisfied constraints $s_\mathrm{min}$.
        We stop the chains once $s_\mathrm{min}\geq\langle s\rangle_\mathrm{DQI}$, and we refer to the required number of steps as the maximum mixing time $\tau_\mathrm{max}$ (see Fig.~\ref{fig:tau_search_scaling}b for further details).
        d) As the problem size increases, the mixing time $\tau_\mathrm{max}$ is well-approximated by an exponential $1.1^{n_p}$, which corresponds to a linear fit on a semi-log plot (we ignore the first point, as an outlier).
        Here $n_p$ is the equivalent number of qubits corresponding to OPI for prime $p$.
        The largest OPI instance we consider corresponds to $p=53$, with $n_p=156$.
        We observe that the qualitative behavior of the average mixing time for OPI would likely not differ much from that of $\tau_\mathrm{max}$.
        Refer to Table~\ref{tab:fits} for the precise numerical values of all fits.
    }
    \label{fig:tau_search}
\end{figure*}

\paragraph{Search.}
For max-XORSAT, the results in Fig.~\ref{fig:tau_search}a--b show that the average mixing time $\tau_{\mathrm{avg}}$ scales polynomially in $n$ (as approximately $n^5$), demonstrating that block Gibbs sampling can efficiently solve the max-XORSAT optimization problem at the DQI performance level.
For OPI, the results in Fig.~\ref{fig:tau_search}c--d show a striking contrast with max-XORSAT: the maximum mixing time $\tau_{\mathrm{max}}$ scales \emph{exponentially} in $n_p$, growing approximately as $1.1^{n_p}$.  
We observe that the average mixing time displays a qualitatively-similar behavior.
Further details on the search experiments can be found in Fig.~\ref{fig:tau_search_scaling}.

\begin{figure*}[t]
    \centering
    \includegraphics{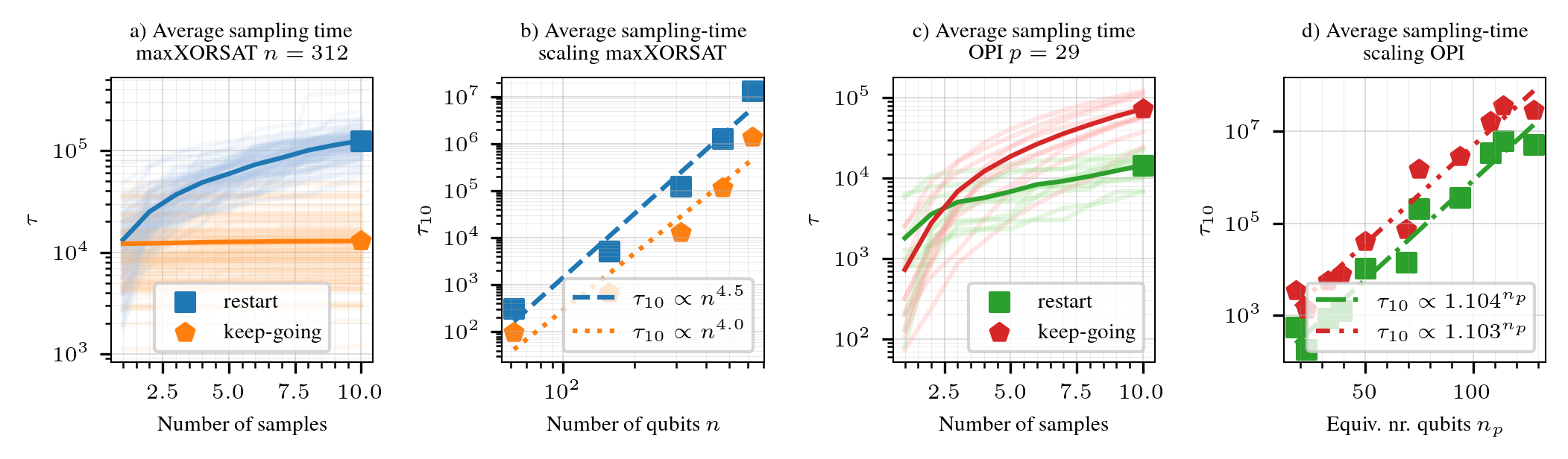}
    \caption{
        \textbf{Scaling behavior of Markov-chain sampling times.}
        a) For max-XORSAT, we sample $100$ uniform random vectors and employ both sampling algorithms described in Section~\ref{ss:mcmc-optim-samp}.
        We report the average number of steps required to produce up to $N=10$ different samples (see Fig.~\ref{fig:tau_sampling_scaling_max-XORSAT} for further details).
        b) As the problem size increases, the average sampling time $\tau_{10}$ is well-approximated by a power law of degree $4$ or $4.5$.
        c) For OPI, we sample $10$ uniform random sets and employ both sampling algorithms.
        We report the average number of steps required to produce up to $N=10$ different samples (see Fig.~\ref{fig:tau_sampling_scaling_OPI} for further details).
        d) As the problem size increases, the average sampling time $\tau_{10}$ is well-approximated by an exponential $1.1^{n_p}$.
        Refer to Table~\ref{tab:fits} for the precise numerical values of all fits.
    }
    \label{fig:tau_sampling}
\end{figure*}

\paragraph{Sampling.}
For max-XORSAT, Figure~\ref{fig:tau_sampling}a--b shows that the average sampling time $\tau_{10}$ for both sampling algorithms has the same qualitative scaling as the average mixing time, namely it is well-approximated by a power law.
We observe a factor $\approx10$ multiplicative advantage for keep-going over restart.
Further details on the sampling experiments for max-XORSAT can be found in Fig.~\ref{fig:tau_sampling_scaling_max-XORSAT}.
For OPI, Fig.~\ref{fig:tau_sampling}c--d are similar to the case of max-XORSAT: both sampling algorithms display the same qualitative asymptotic scaling in $n_p$ as the maximum mixing time for search, yet this time \emph{keep going} is \emph{slower} than \emph{restart} by a factor $\approx10$.
Further details on the sampling experiments for OPI can be found in Fig.~\ref{fig:tau_sampling_scaling_OPI}.

\paragraph{Lower thresholds for OPI.}

\begin{figure*}[t]
    \centering
    \includegraphics{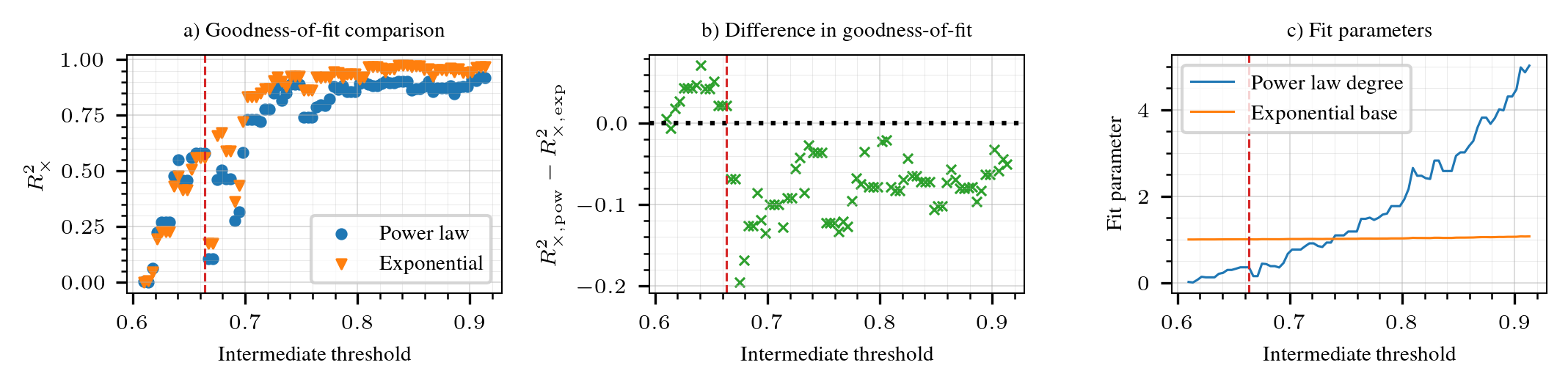}
    \caption{
        \textbf{Comparison of goodness-of-fit for lower score thresholds.}
        The performance guarantees of all known polynomial-time \emph{classical} algorithms yield thresholds ranging between $0.55+o(1)$ and $0.75$~\cite{Jordan2025optimization}, whereas DQI for these parameter regimes of OPI yields a threshold exceeding $0.9$.
        a) Goodness-of-fit $R^2_\times$ for different score thresholds, for both a power law (blue dots) and an exponential fit (orange triangles).
        We observe a similar behavior for both fits in larger thresholds, and no clear pattern around the phase transition (red dashed line).
        Low values of $R^2_\times$, say $R^2_\times<0.8$ are indicative of a poor fit.
        b) Difference in goodness-of-fits, where the dashed red line captures the phase transition, from a positive difference (signaling that a power law is preferable) to a negative one (exponential is a better fit).
        Together with panel a), we identify the phase transition to be in the regime where both fits achieve low $R^2_\times$.
        c) Values of the fit parameters for different thresholds.
        The degree of the power-law fit around the phase transition is below $1$, which suggests that this threshold is within reach of random guessing with several repetitions.
        Beyond the phase transition, we observe that the degree of the power law increases steadily with the value of the threshold, roughly indicating an unbounded degree, and hence an exponential underlying behavior.
        The base of the exponential remains close to $1$ throughout.
    }
    \label{fig:comparison_thresholds}
\end{figure*}

    One additional difference between max-XORSAT and OPI is the numerical value of the threshold $\sDQI/m$.
    The DQI performance threshold for OPI is $\gtrsim0.9$, approximately $0.1$ higher than the corresponding threshold for max-XORSAT $(\gtrsim0.8)$.
    In both cases, we expect random guessing to achieve an approximation ratio of $0.5$, and slightly higher for certain variants of simulated annealing or the so-called Prange algorithm~\cite{Jordan2025optimization,khattar2025dqi_opi}.
    We investigate the effect of setting a lower threshold for OPI on the empirical runtime scaling of our MCMC approaches.
    For this, we can simply re-interpret the data we already obtained for the search problem (as depicted in Figs.~\ref{fig:tau_search}b and~\ref{fig:tau_search_scaling}b).

    We collected the runtime scaling for a range of intermediate thresholds, and for each we fit both a power law and an exponential, as shown in Fig.~\ref{fig:comparison_thresholds}.
    The quantity of interest is the \emph{goodness-of-fit} $R_\times^2$ (which we introduce in Appendix~\ref{a:further_details}).
    While both fits are similarly-good, we observe a phase transition: a power law is a better fit for thresholds below $0.66$, and an exponential fit is better above this value.
    In comparison, the expected performance of Prange's algorithm in these instances is $0.75$, above the threshold where the exponential fit is better than the power-law one.
 
\subsection{Discussion}\label{ss:discussion}

    A key takeaway from our results is that using classical \emph{sampling} algorithms to solve \emph{optimization} problems produces competitive results.
    In familiar terms from quantum computing, we have simulated the output distribution of the DQI state vector $\DQI$ without simulating the quantum circuit that prepares it step-by-step.
    We close this section by briefly interpreting and discussing our empirical results.
    
    Regarding the search problem, our numerical results are fully consistent with the existing claims in the literature~\cite{Jordan2025optimization}.
    The empirical runtime of our MCMC methods for max-XORSAT is polynomial in the number of qubits, thus matching the qualitative behavior of common classical optimization algorithms.
    The empirical runtime we observe for OPI is exponential in the number of qubits, hence not challenging the claims for super-polynomial quantum advantage from DQI for this problem.
    Still, the basis of the exponential function we obtain is rather small.
    Assuming the observed behavior persists beyond the instance sizes we considered, our results predict that classical algorithms may be able to match the performance of DQI for OPI \emph{in practice} for a wide range of instance sizes.

    We note that the qualitative runtime scaling we observe for the search problem is roughly maintained also for the sampling problem.
    In this sense, we take our results as empirical evidence that the sampling version of these optimization problems is not arbitrarily harder than their search version.
    Beyond the qualitative asymptotic scaling, the main difference between our sampling experiments for max-XORSAT and OPI is which of the two sampling algorithms performs better.
    We see that \emph{keep-going} is faster for max-XORSAT and \emph{restart} is faster for OPI.
    We interpret this in terms of the overlap gap property discussed above~\cite{anschuetz2025decoded}: the solutions to max-XORSAT seem to cluster together, and the same may not be true for OPI.
    This may make intuitive sense just from sparsity considerations alone: the ensemble from which we draw the coefficient matrix $B$ for max-XORSAT ensures that few constraints are affected by a single bitflip, so the score $f(x)$ is stable.
    Conversely, all entries of the Vandermonde matrix we have in OPI are non-$0$, thus yielding a score function $f(x)$ that is highly unstable, even under small perturbations of $x$.
    
    We make one final comment on the possibility of \emph{warm-starting} our sampling algorithms.
    Using the language of Section~\ref{ss:mcmc-optim-samp}, suppose the Markov chain is such that the \emph{mixing time} is much larger than the \emph{intra-sample} time $\tau\gg I$.
    Then, our \emph{keep-going} algorithm could be modified to start from an arbitrary good sample, potentially one produced by DQI.
    We envision a situation where the cost of producing a single sample via DQI, $T_{\operatorname{DQI}}$, is higher than the intra-sample time but lower than the mixing time of our Markov chain: $I<T_{\operatorname{DQI}}<\tau$.
    In this situation, the best choice would be to run DQI once to produce an initial good sample, and take that sample as the initial state of the Markov chain.
    Characterizing such situations in practice is an exciting open problem.

    Our preliminary results on lowering the threshold for OPI remain inconclusive.
    Drawing strong conclusions would likely require reaching larger instance sizes, which ought to separate better the goodness-of-fit of both approaches.
    We interpret these results as evidence that the exponential runtime we observe in Fig.~\ref{fig:tau_search} is robust under perturbations of the target score.

\section{Conclusions and outlook
\label{sec:conclusion}}

    In this work, we have shed further light on the complexity of \emph{decoded quantum interferometry} (DQI), both theoretically and empirically.
    Our simplified analysis of the performance of DQI allowed us to decouple the three possible versions of the optimization problem: decision, search, and sampling.
    We have shown that studying the potential quantum advantage of DQI requires going beyond the decision problem, and we find that studying the hardness of sampling may require different tools than in previous quantum sampling settings.
    We have proposed to use classical approximate sampling algorithms to tackle the optimization tasks, specifically a usual variant of \emph{Markov chain Monte Carlo} (MCMC).
    We have performed several scaling experiments, where we report the number of steps necessary for MCMC methods to match the performance of DQI.
    Reaching beyond $100$ qubits, we have observed run-times consistent with previous claims of quantum speed-ups.
    In particular, the number of steps necessary for MCMC methods to match the performance of DQI displays a qualitative scaling of $1.096^{n_p}$, where $n_p$ is the number of qubits of the output state of DQI.
    Our combined results call for a nuanced study of the potential quantum advantage for DQI.

    Our goal was to compare usual MCMC approaches against the expected performance of DQI.
    In parallel, it would be interesting to perform similar scaling experiments using the best known classical heuristic algorithms, like Prange's algorithm and simulated annealing, as discussed in Refs.~\cite{Jordan2025optimization,khattar2025dqi_opi,shutty2026optimization}.
    Similarly, our experiments could be extended with an emphasis on testing our MCMC methods in terms of their per-instance sensitivity to different random seeds.

    We identify several avenues for progress beyond our work, in part enabled by our public repository~\cite{github_repository}.
    We expect interesting behavior to arise from fine-tuning the polynomial degree of the distribution.
    In our experiments, we set the polynomial degree to be the same as in the original DQI algorithm~\cite{Jordan2025optimization}.
    Since then, several improvements have been proposed, using quantum decoders~\cite{chailloux2024softdecoders,chailloux2025opi,shutty2026optimization,horinaga2026worst}, which would correspond to larger numbers of decodable errors.
    Further, specifically for MCMC methods, there may be a trade-off between the expected score and the shape of the score landscape (the higher the degree, the sharper the landscape), which may result in longer mixing times.
    In particular, it would be interesting to measure the effect of the choice of polynomial degree on the basis of the exponential fit for the mixing time of the Markov chain for OPI.
    Beyond extending our studies to other optimization problems, a promising avenue would be to extend our methods to other sampling strategies.
    Indeed, since we focused on off-the-shelf MCMC methods, it may be possible to improve the runtime by fine-tuning the sampler.
    An open question that arises from the combination of our analytical and numerical results is the potential benefit of using the marginals of the DQI distribution: we may be able to sample from low-body marginals classically efficiently, yet it may be the case that they are not helpful to solve the optimization problem.
    Finally, we highlight the possibility of extending our framework to Hamiltonian DQI~\cite{schmidhuber2025hamiltonian_dqi}.

More broadly, we believe that the combination of complexity-theoretic analysis and large-scale empirical benchmarking developed here provides a useful framework for assessing the prospects of DQI and related quantum optimization algorithms. While our results clarify several aspects of the computational landscape, they also expose fundamental open questions concerning the interplay between sampling complexity and optimization performance. It is our hope that pursuing this direction will ultimately lead to a deeper understanding of the extent to which quantum computers can outperform classical methods on practically relevant optimization problems.

\subsection*{Acknowledgements}
The authors would like to thank Maximilian J. Kramer, Louis Schatzki, and David Sutter for insightful discussions.
E.G.-F. is a 2023 Google PhD Fellowship recipient.
E.G.-F., L.B., and J.E. have been supported by the German Federal Ministry of Research, Technology and Space BMFTR (MuniQC-Atoms, QuSol, Hybrid++, PasQuops, PraktiQOM), Clusters of Excellence (ML4Q, MATH+), the QuantERA, the Munich Quantum Valley, Berlin Quantum, the Quantum Flagship (Millenion, Pasquans2), the DFG (CRC 183, SPP 2514), and the European Research Council (DebuQC).

IBM Bob and GPT‑5.6 were used to assist with code development and mathematical consistency checks during the preparation of this manuscript.
All scientific ideas, analyses, and conclusions are those of the authors, and all AI-assisted code and text were carefully reviewed, manually checked, and validated prior to inclusion.

\bibliography{references}

\onecolumngrid
\clearpage
\appendix

\numberwithin{theorem}{section}
\numberwithin{algorithm}{section}
\numberwithin{figure}{section}
\numberwithin{table}{section}

\begin{center}
\large{Supplementary material for \\ ``Approximate sampling from decoded quantum interferometry via Markov chain Monte Carlo methods''
}
\end{center}

\section{Construction of the optimal polynomial} \label{a:optimal_polynomial}

Let $f=f_1+\ldots+f_m$ be the objective function, with $f_i:\mathbb{F}_p\rightarrow\{+1,-1\}$, $r:=|f_i^{-1}(+1)|\in\{0,\ldots,p-1\}$ and
\begin{equation}
    P(f)=\sum_{k=0}^\ell \alpha_k f^k.
\end{equation}
In Ref.~\cite{Jordan2025optimization}, the authors show how to compute the optimal polynomial $P$ maximizing $\langle s\rangle$. 
The polynomial is not expressed directly in terms of the coefficients $\alpha_k$ but rather as
\begin{equation}
    P(f) = \sum_{k=0}^{\ell} u_k P^{(k)}\bigl(g_1(f_1), \ldots, g_m(f_m)\bigr),
\end{equation}
where
\begin{itemize}
    \item $P^{(k)}(x_1,\ldots,x_m)= \sum_{1 \le a_1 < \cdots < a_k \le m}
    x_{a_1} x_{a_2} \cdots x_{a_k}$ is the degree-$k$ elementary symmetric polynomial,
    \item $u_k = w_k\left/\sqrt{p^{n-k}\binom{m}{k}}\right.$,
    \item $g_i(x_i) := \frac{f_i(x_i) - \bar f}{\varphi} \in \left\{ \frac{1-\bar f}{\varphi}, \frac{-1-\bar f}{\varphi} \right\}$, where $\bar f = \frac{2r}{p} - 1$, and $\varphi = \sqrt{4r\left(1-\frac{r}{p}\right)}$.
\end{itemize}

The optimal polynomial $P$ is obtained by choosing $\vec w = (w_0,\ldots,w_\ell)^{\mathsf T} $ to be the eigenvector corresponding to the largest eigenvalue of the tridiagonal matrix
\begin{equation}
    A^{(m,\ell,\delta)} =
\begin{pmatrix}
0   & a_1 & 0   & \cdots & 0 \\
a_1 & \delta   & a_2 & \ddots & \vdots \\
0   & a_2 & 2\delta   & \ddots & 0 \\
\vdots & \ddots & \ddots & \ddots & a_\ell \\
0 & \cdots & 0 & a_\ell & \ell \delta
\end{pmatrix},
\end{equation}
where 
$a_k = \sqrt{k(m-k+1)}$ and 
\begin{equation}
\delta = \frac{p-2r}{\sqrt{r(p-r)}}.
\end{equation}

\section{A simplified performance analysis for DQI}\label{a:simplified_analysis}

    In this section we first relate the DQI distribution to the simpler uniform distribution.
    Then, we show that the uniform distribution is easy to characterize up to a level that relates to the error-correcting code that DQI tackles in solving the optimization problem.
    Recall the polynomial $P(f)=\sum_{k=0}^\ell\alpha_kf^k$ from Appendix~\ref{a:optimal_polynomial}.
    The DQI distribution over symbol strings $x\in\bbF_p^n$ is given by
    \begin{align}
        \PDQI(x) &= P^2(f(x)) = \sum_{k,k'=0}^\ell \alpha_k\alpha_{k'}f^{k+k'}(x)= \sum_{k=0}^{2\ell}\beta_kf^k(x),
    \end{align}
    where the $\beta_k$ are auxiliary coefficients and $\alpha_k$ are chosen so that $\sum_{x\in\bbF_p^n}\PDQI(x)=1$.
    Within the scope of DQI, we further assume that $\ell\leq\lfloor (d-1)/2\rfloor$, namely the degree of the polynomial is at most half the distance of the error-correcting code corresponding to the optimization problem.

    \begin{customlemma}{\ref{l:moments_dqi_unif}}[Moments of the DQI vs uniform distribution]
        Suppose we use DQI with a polynomial of degree $\ell$.
        Then, the $k^\text{th}$ moment of the objective function $f(x)$ under the DQI distribution is a linear combination of moments under the uniform distribution as
        \begin{align}
            \bbE_{x\sim \PDQI}\left[f^k(x)\right] &= p^n \sum_{k'=0}^{2\ell}\beta_{k'}\bbE_{x\sim\operatorname{Unif}}\left[f^{k+k'}(x)\right].
        \end{align}
    \end{customlemma}
    \begin{proof}
        We prove the statement directly.
        We start by anchoring what the moments of the uniform distribution look like, where $\calP_{\operatorname{Unif}}(x)=1/p^n$ for $x\in\bbF_p^n$. We find
        \begin{align}
            \bbE_{x\sim\operatorname{Unif}}\left[f^{k'}(x)\right] &= \sum_{x\in\bbF_p^n}\calP_{\operatorname{Unif}}(x) f^{k'}(x) = \sum_{x\in\bbF_p^n} \frac{1}{p^n} f^{k'}(x), \\
             \sum_{x\in\bbF_p^n}f^{k'}(x) &= p^n \bbE_{x\sim\operatorname{Unif}}\left[f^{k'}(x)\right].
        \end{align}
        Then, we manipulate the moments of the DQI distribution until they match the desired form
        \begin{align}
            \bbE_{x\sim \PDQI}\left[f^k(x)\right] &= \sum_{x\in\bbF_p^n} \PDQI(x) f^k(x) 
            = \sum_{x\in\bbF_p^n} \left(\sum_{k'=0}^{2\ell} \beta_{k'} f^{k'}(x)\right)f^k(x) 
            = \sum_{x\in\bbF_p^n}\sum_{k'=0}^{2\ell} \beta_{k'} f^{k+k'}(x) \\
            &= \sum_{k'=0}^{2\ell} \beta_{k'} \left(\sum_{x\in\bbF_p^n}f^{k+k'}(x)\right)
            = \sum_{k'=0}^{2\ell} \beta_{k'} p^n\bbE_{x\sim\operatorname{Unif}}\left[f^{k+k'}(x)\right]
            = p^n \sum_{k'=0}^{2\ell} \beta_{k'} \bbE_{x\sim\operatorname{Unif}}\left[f^{k+k'}(x)\right].
        \end{align}
    \end{proof}

    \begin{remark}[The $p^n$ factor:]
        Due to normalization, the coefficients $\beta_k$ scale as $1/p^n$, countering the extra factor that appears in Lemma~\ref{l:moments_dqi_unif}.
    \end{remark}

    \begin{lemma}[The first $d$ moments of the uniform distribution -- max-XORSAT]\label{l:distance_uniform_distribution_moments_max-XORSAT}
        Consider the case $p=2$, $r=1$.
        For all $k<d$, the $k^\text{th}$ moment of the objective function $f(x)$ under the uniform distribution fulfills $\bbE_{x\sim\operatorname{Unif}}\left[f^k(x)\right] = \bbE_{J\sim\operatorname{Binom}(m,1/2)}\left[(2J-m)^{k}\right]$.
        In particular, the odd moments are $0$.
    \end{lemma}
    As we write below, Lemma~\ref{l:distance_uniform_distribution_moments_max-XORSAT} is a special case of the more general Theorem~\ref{thm:distance_uniform_distribution_moments_maxLINSAT}.
    Nevertheless, we present it as an independent result, since the proof is considerably simpler.
    \begin{proof}[Proof of Lemma~\ref{l:distance_uniform_distribution_moments_max-XORSAT}]
        We prove this statement directly.
        We first express the moment as a sum over possible codewords of length up to $k$, and next observe that the only codeword of length less than $d$ is the all-$0$ codeword.
        We start by  inserting the definition of the objective function for max-XORSAT, to get
    \begin{align}
            \bbE_{x\sim\operatorname{Unif}}\left[f^k(x)\right] &= \sum_{x\in\bbF_2^n} \calP_{\operatorname{Unif}}(x) f^k(x)
             = \sum_{x\in\bbF_2^n}\frac{1}{2^n} \left(\sum_{i=1}^m (-1)^{\langle b_i,x\rangle + v_i}\right)^k.
        \end{align}
        As in the common definition of $f$, $b_i$ correspond to the rows of the coefficient matrix $B$.
        We expand the $k^\text{th}$ power as a sum with a multi-index $y=(y_i)_{i=1}^m$, with $y_i\in[k]\coloneqq\{0,\ldots,k\}$, and $\sum_{i=1}^my_i=k$,
        \begin{align}
            2^n \bbE_{x\sim\operatorname{Unif}}\left[f^k(x)\right] &= \sum_{x\in\bbF_2^n} \left(\sum_{i=1}^m (-1)^{\langle b_i,x\rangle + v_i}\right)^k
            = \sum_{x\in\bbF_2^n} \sum_{y_1+\cdots+y_m=k}\binom{k}{y_1,\ldots,y_m} (-1)^{\langle B^\intercal y,x\rangle + \langle y,v\rangle} \\
            &= \sum_{y_1+\cdots+y_m=k}\binom{k}{y_1,\ldots,y_m} (-1)^{\langle y,v\rangle}\sum_{x\in\bbF_2^n}(-1)^{\langle B^\intercal y,x\rangle} \\
            &= \sum_{y_1+\cdots+y_m=k}\binom{k}{y_1,\ldots,y_m} (-1)^{\langle y,v\rangle} 2^n\delta_{B^\intercal y\bmod{2},0}.
        \end{align}
        In the first line, we have replaced the sum over $i$ as an inner product with the multi-index $y$, using the multinomial theorem.
        This includes identifying $\langle b_i,x\rangle=[Bx]_i$, as well as $\langle y,Bx\rangle=\langle B^\intercal y,x\rangle$.
        On the last line, we have used the usual Fourier convolution identity $\sum_{x\in\bbF_2^n}(-1)^{\langle a,x\rangle}=2^n\delta_{a\bmod{2},0}$.
        The notation $\delta_{a,b}$ refers to the Kronecker delta, which equals $1$ if $a=b$, and $0$, otherwise.
        Visual inspection of the identity
        \begin{align}
            \bbE_{x\sim\operatorname{Unif}}\left[f^k(x)\right] &= \sum_{y_1+\cdots+y_m=k}\binom{k}{y_1,\ldots,y_m} (-1)^{\langle y,v\rangle} \delta_{B^\intercal y\bmod{2},0}
        \end{align}
        reveals that this is a signed sum over all $k$-length codewords in the error-correction code specified by the parity-check matrix $B^\intercal$.
        Here, we use the assumption that this error-correction code has distance $d$, which per definition is a lower bound on the Hamming weight of any nonzero codeword.
        In particular, this means that, for all $y$ with $\lvert y\rvert <d$, we have $\delta_{B^\intercal y\bmod{2},0}=\delta_{y\bmod{2},0}$.
        The Kronecker delta $\delta_{y\bmod{2},0}$ acts as a selector function, and it discriminates between two main cases, namely whether $k$ is even or odd.

        If $k$ is odd, we can rewrite $k=2a+1$, for $a\in\bbN$.
        It is not possible to fulfill $\lvert y\rvert=2a+1$ and $y\bmod{2}=0$ simultaneously, so all terms in the sum are $0$, and hence $\bbE_{x\sim\operatorname{Unif}}[f^{2a+1}(x)]=0$ for all $a<(d-1)/2$.
        For $k$ even, we can rewrite this as $k=2a$, for $a\in\bbN$.
        Then, $y\bmod{2}=0$ implies that not only the sum $\sum_i y_i=2a$ is even, but also that each term $y_i$ is even,
        \begin{align}
            \bbE_{x\sim\operatorname{Unif}}\left[f^{2a}(x)\right] &= \sum_{y_1+\cdots+y_m=2a}\binom{2a}{y_1,\ldots,y_m} (-1)^{\langle y,v\rangle} \delta_{y\bmod{2},0} \\
            &= \sum_{y_1+\cdots+y_m=2a}\binom{2a}{y_1,\ldots,y_m} (-1)^{\langle y,v\rangle} \prod_{i=1}^m\delta_{y_i\bmod{2},0} \\
            &= \sum_{y_1+\cdots+y_m=2a}\binom{2a}{y_1,\ldots,y_m} \prod_{i=1}^m\delta_{y_i\bmod{2},0}.
        \end{align}
        Since all entries of $y$ must be even, it follows that $(-1)^{\langle y,v\rangle}=1$.
        We next rewrite the Kronecker delta using the  identity $\delta_{y_i\bmod{2},0} = \frac{1}{2}\sum_{\varepsilon_i\in\{\pm1\}}\varepsilon_i^{y_i}$.
        The variables $\varepsilon_i$ are usually referred to as \emph{Rademacher} variables.
        Note that the product over $i\in[m]$ and the sum over $m$ Rademacher variables $\varepsilon\in\{\pm1\}^m$ commute, since 
        \begin{align}
            \prod_{i=1}^m\delta_{y_i\bmod{2},0} &= \prod_{i=1}^m\left(\frac{1}{2}\sum_{\varepsilon_i\in\{\pm1\}}\varepsilon_i^{y_i}\right) = \frac{1}{2^m}\sum_{\varepsilon\in\{\pm1\}^m}\prod_{i=1}^m\varepsilon_i^{y_i}.
        \end{align}
        With this, we can re-write the above expression as an average 
    \begin{align}
            \bbE_{x\sim\operatorname{Unif}}\left[f^{2a}(x)\right] &= \sum_{y_1+\cdots+y_m=2a}\binom{2a}{y_1,\ldots,y_m} \left(\frac{1}{2^m}\sum_{\varepsilon\in\{\pm1\}^m}\prod_{i=1}^m\varepsilon_i^{y_i}\right) \\
            &= \frac{1}{2^m}\sum_{\varepsilon\in\{\pm1\}^m} \sum_{y_1+\cdots+y_m=2a}\binom{2a}{y_1,\ldots,y_m} \prod_{i=1}^m\varepsilon_i^{y_i}
        \end{align}
        over Rademacher variables.
        We can re-interpret the inner sum via the multinomial theorem, to get
        \begin{align}
            \bbE_{x\sim\operatorname{Unif}}\left[f^{2a}(x)\right] &= \frac{1}{2^m}\sum_{\varepsilon\in\{\pm1\}^m} \sum_{y_1+\cdots+y_m=2a}\binom{2a}{y_1,\ldots,y_m} \prod_{i=1}^m\varepsilon_i^{y_i} = \frac{1}{2^m}\sum_{\varepsilon\in\{\pm1\}^m}\left(\sum_{i=1}^m \varepsilon_i\right)^{2a}.
        \end{align}
        The sum $\varepsilon_1+\cdots+\varepsilon_m$ only depends on how many of the variables take value $-1$.
        Denoting the number of $+1$ variables by $j\in\{0,\ldots,m\}$, the sum evaluates to $\varepsilon_1+\cdots+\varepsilon_m=2j-m$.
        We can re-write the sum over $\varepsilon$ as a sum 
\begin{align}
            \bbE_{x\sim\operatorname{Unif}}\left[f^{2a}(x)\right] &= \frac{1}{2^m}\sum_{\varepsilon\in\{\pm1\}^m}\left(\sum_{i=1}^m \varepsilon_i\right)^{2a} = \frac{1}{2^m}\sum_{j=0}^m\binom{m}{j}(2j-m)^{2a}
        \end{align}
       over $j$, accounting for duplicates.
        As a last step, we introduce a binomial-distributed random variable $J\sim\operatorname{Binom}(m,1/2)$, with $\bbP_{\operatorname{Binom}}[J=j]=\binom{m}{j}/2^m$, and we observe that the expression we obtained is exactly the expected value of $(2J-m)^{2a}$:
        \begin{align}
            \bbE_{x\sim\operatorname{Unif}}\left[f^{2a}(x)\right] &= \frac{1}{2^m}\sum_{j=0}^m\binom{m}{j}(2j-m)^{2a} = \sum_{j=0}^m\bbP_{\operatorname{Binom}}[J=j](2j-m)^{2a} = \bbE_{J\sim\operatorname{Binom}(m,1/2)}\left[(2J-m)^{2a}\right].
        \end{align}
        Substituting $k=2a$, we recover the claimed result.
    \end{proof}
    
    \begin{customthm}{\ref{thm:distance_uniform_distribution_moments_maxLINSAT}}[The first $d$ moments of the uniform distribution -- max-LINSAT]
        Suppose a max-LINSAT instance whose dual code has distance $d$.
        Let $J\sim\operatorname{Binom}(m, r/p)$ be a binomial-distributed random variable.
        Then, for all $k<d$, the $k^\text{th}$ moment of the objective function $f(x)$ under the uniform distribution fulfills
        \begin{align}
            \bbE_{x\sim\operatorname{Unif}}[f^k(x)] &= \bbE_{J}\left[(2J-m)^k\right].
        \end{align}
    \end{customthm}
    \begin{proof}
        We have $f(x)=\sum_i f_i(\langle b_i,x\rangle)$, with $b_i$ being the $i^\text{th}$ row of $B$: $\langle b_i,x\rangle=[Bx]_i$, and where 
    \begin{align}
        f_i(a)=\begin{cases}
            1& a\in F_i,\\
            -1  & a\notin F_i.
        \end{cases}
    \end{align}
    We take the sets $F_i$ all of the same size $r=\lvert F_i\rvert$.
    From this, it follows that $\sum_{a\in\bbF_p} f_i(a) = 2r-p$ for all $i\in[m]$.
    The (inverse) Fourier transform on $\bbF_p$ for the functions $f_i$ is given by
    \begin{align}
        \tf_i(y) &=\frac{1}{\sqrt{p}}\sum_{a\in\bbF_p}\omega_p^{-ay}f_i(a), \\
        f_i(a) &= \frac{1}{\sqrt{p}}\sum_{y\in\bbF_p} \omega_p^{ay}\tf_i(y),
    \end{align}
    where $\omega_p=e^{2\pi i/p}$ is a $p^\text{th}$ root of unity $\omega_p^p=1$.
    With these, we can express $f(x)$ in terms of the Fourier transforms 
    \begin{align}
        f(x)&= \sum_{i=1}^m f_i(\langle b_i,x\rangle)
        = \sum_{i=1}^m \left(\frac{1}{\sqrt{p}}\sum_{y\in\bbF_p} \omega_p^{\langle y\,b_i, x\rangle} \tf_i(y)\right)
        = \frac{1}{\sqrt{p}}\sum_{i=1}^m \sum_{y\in\bbF_p} \omega_p^{\langle y\,b_i, x\rangle} \tf_i(y)
    \end{align}
    of the individual terms $\tf_i$,
    and similarly for the $k^\text{th}$ power of the objective function, where we re-defined $i$ as a length-$k$ multi-index $i\in[m]^k$,
    \begin{align}
        f^k(x) &= \left(\sum_{i_0=1}^m f_{i_0}(\langle b_{i_0},x\rangle)\right)^k 
        =\left(\sum_{i_1=1}^m f_{i_1}(\langle b_{i_1},x\rangle)\right)\cdots\left(\sum_{i_k=1}^m f_{i_k}(\langle b_{i_k},x\rangle)\right)
        = \sum_{i\in[m]^k} \prod_{j=1}^k f_{i_j}(\langle b_{i_j},x\rangle) \\
        &= \sum_{i\in[m]^k} \prod_{j=1}^k \left(\frac{1}{\sqrt{p}} \sum_{y_j\in\bbF_p}\omega_p^{\langle y_j b_{i_j}, x\rangle} \tf_{i_j}(y_j)\right)
        =\frac{1}{\sqrt{p^{k}}}\sum_{i\in[m]^k}\sum_{y\in\bbF_p^k} \omega_p^{\left\langle\sum_{j=1}^k y_j b_{i_j}, x\right\rangle}\prod_{j=1}^k \tf_{i_j}(y_j).
    \end{align}

    For completeness: we denote $y_j\in\bbF_p$ as a scalar and $y\in\bbF_p^k$ as a vector $y = (y_{j})_{j=1}^k$ of dimension $k$.
    Note that the product $y_j b_{i_j}$ is a scalar-vector product, where $y_j\in\bbF_p$ is one of the entries of $y$ and $b_{i_j}$ is one of the rows of $B$, indexed by the $j^\text{th}$ entry of $i\in[m]^k$.
    We are interested in the moments of the uniform distribution given by
    \begin{align}
        \bbE_{x\sim\operatorname{Unif}}[f^k(x)] &= \sum_{x\in\bbF_p^n} \calP_{\operatorname{Unif}}(x) f^k(x) = \sum_{x\in\bbF_p^n} \frac{1}{p^n} f^k(x) = \frac{1}{p^n}\sum_{x\in\bbF_p^n} f^k(x) \\
        &= \frac{1}{p^n} \sum_{x\in\bbF_p^n} \left(\frac{1}{\sqrt{p^{k}}}\sum_{i\in[m]^k}\sum_{y\in\bbF_p^k} \omega_p^{\left\langle\sum_{j=1}^k y_{j} b_{i_j}, x\right\rangle}\prod_{j=1}^k \tf_{i_j}(y_{j})\right) \\
        &= \frac{1}{p^n\sqrt{p^k}} \sum_{i\in[m]^k}\sum_{y\in\bbF_p^k} \underbrace{\sum_{x\in\bbF_p^n}\omega_p^{\left\langle\sum_{j=1}^k y_{j} b_{i_j}, x\right\rangle}}_{p^n \delta_{\sum_{j=1}^k y_{j} b_{i_j}, 0 \bmod p} = p^n \delta_{B^\intercal\left(\sum_{j=1}^ky_j e_{i_j}\right),0\bmod p}}\prod_{j=1}^k \tf_{i_j}(y_{j}) \\
        &= \frac{1}{\sqrt{p^k}}\sum_{i\in[m]^k}\sum_{y\in\bbF_p^k} \delta_{B^\intercal\left(\sum_{j=1}^ky_j e_{i_j}\right),0\bmod p} \prod_{j=1}^k \tf_{i_j}(y_j).
    \end{align}
    In the second-to-last line, we introduced the elementary basis vectors $e_{i_j}$.
    Under the assumption that $k<d$, it follows that the only choices for $y$ which fulfill the condition of the Kronecker delta also fulfill a simpler condition, namely $\delta_{B^\intercal\left(\sum_{j=1}^ky_j e_{i_j}\right),0\bmod p} = \delta_{\left(\sum_{j=1}^ky_j e_{i_j}\right),0\bmod p}$.
    In other words, the only codeword of length $\leq k<d$ is the trivial all-$0$ codeword.
    We take this condition into the specification of the sum, to arrive at
    \begin{align}
        \bbE_{x\sim\operatorname{Unif}}[f^k(x)] &=\cdots = \frac{1}{\sqrt{p^k}}\sum_{i\in[m]^k}\sum_{\substack{y\in\bbF_p^k\\\sum_{j=1}^k y_j e_{i_j}=0\bmod p}} \prod_{j=1}^k \tf_{i_j}(y_j) \coloneqq \frac{1}{\sqrt{p^k}}\sum_{i\in[m]^k} F(i).
    \end{align}
    In this step,
    we have renamed the terms inside the first sum as a function of the multi-index $i$.
    Let $\calS_k$ be the symmetric group of $k$ elements.
    We note from visual inspection that $F$ is invariant under permutations: for any $i\in[m]^k$, and for any $\sigma\in\calS_k$, it holds $F(i)=F(\sigma(i))$, with $\sigma(i)\coloneqq(i_{\sigma(1)}, \ldots, i_{\sigma(k)})$.
    From this, it follows that $F$ only depends on how many entries of $i$ take each value.
    Let $n_1,\ldots,n_m\in\bbN$, with $\sum_{j=1}^mn_j=k$, and denote as $1^{n_1}2^{n_2}\cdots m^{n_m}$ a length-$k$ vector whose entries are ordered: $n_1$ many $1$s, $n_2$ many $2$s, and so on.
    Then, for any $i\in[m]^k$, there exist values $(n_1,\ldots,n_m)$ and a permutation $\sigma\in\calS_k$ such that $i=\sigma(1^{n_1}2^{n_2}\cdots m^{n_m})$.
    We abuse notation to highlight this fact, and write $F(i)\coloneqq  F(n_1,\ldots,n_m)$.
    In order to compute $\sum_{i\in[m]^k}F(i)$, we must then compute $F(n_1,\ldots,n_m)$, and account for the number of permutations that produce different results.

    We take on the first task, to get
\begin{align}
        F(n_1,\ldots,n_m) &\coloneqq \sum_{\substack{y\in\bbF_p^k\\\sum_{j=1}^k y_j e_{i_j}=0\bmod p}} \prod_{j=1}^k \tf_{i_j}(y_j) 
        = \prod_{j=1}^m\sum_{\substack{y\in\bbF_p^{n_j}\\\sum_{t=1}^{n_j}y_t=0\bmod p}} \prod_{t=1}^{n_j}\tf_j(y_t).
    \end{align}
    Note that we re-defined the summation index in the second equality.
    Let us consider an individual term of the multiplication.
    Assuming $n_j\geq2$, we exploit the fact that $\sum_{t=1}^{n_j}y_t=0\bmod p$ is equivalent to $\sum_{t=1}^{n_j-1}y_t=-y_{n_j}\bmod p$, so 
    that
    \begin{align}
        \sum_{\substack{y\in\bbF_p^{n_j}\\
        \sum_{t=1}^{n_j}y_t=0\bmod p}}\prod_{t=1}^{n_j}\tf_j(y_t) &= \sum_{y\in\bbF_p^{n_j-1}}\prod_{t=1}^{n_j-1}\tf_j(y_t)\tf_j\left(-\sum_{t=1}^{n_j-1}y_t\right) \\
        &= \sum_{y\in\bbF_p^{n_j-1}}\left(\prod_{t=1}^{n_j-1}\frac{1}{\sqrt{p}}\sum_{a_t\in\bbF_p}\omega_p^{-a_ty_t}f_j(a_t)\right)\left(\frac{1}{\sqrt{p}}\sum_{a_{n_j}\in\bbF_p}\omega_p^{-a_{n_j}\left(-\sum_{t=1}^{n_j-1}y_t\right)}f_j(a_{n_j})\right) 
        \\
        &= \frac{1}{\sqrt{p^{n_j}}}\sum_{y\in\bbF_p^{n_j-1}}\sum_{a\in\bbF_p^{n_j}}\omega_p^{\sum_{t=1}^{n_j-1}y_t(a_{n_j}-a_t)}\prod_{t=1}^{n_j}f_j(a_t) \\
        &= \frac{1}{\sqrt{p^{n_j}}}\sum_{a\in\bbF_p^{n_j}}\underbrace{\sum_{y\in\bbF_p^{n_j-1}}\omega_p^{\sum_{t=1}^{n_j-1}y_t(a_{n_j}-a_t)}}_{\prod_{j=1}^{n_j-1}\left(p\,\delta_{a_{n_j}-a_t,0}\right)}\prod_{t=1}^{n_j}f_j(a_t)\\
        &= \frac{p^{n_j-1}}{\sqrt{p^{n_j}}}\sum_{a_{n_j}\in\bbF_p}\prod_{t=1}^{n_j}f_j(a_{n_j}) = \frac{p^{n_j-1}}{\sqrt{p^{n_j}}}\sum_{a_{n_j}\in\bbF_p}f_j^{n_j}(a_{n_j}).\label{eq:1/sqrt{p}}
    \end{align}
    The main step in this calculation has been to perform a Fourier convolution.
    We recall that the functions $f_i$ take values $\{\pm1\}$, from which it follows that $f_j^{n_j}(a)=1$ if $n_j$ is even, and $f_j^{n_j}(a)=f_j(a)$ if $n_j$ is odd.
    From the specification of max-LINSAT, we have $\sum_{a\in\bbF_p}f_j(a)=2r-p$, for any $j\in\{1,\ldots,m\}$.
    With this, we conclude
    \begin{align}
        \sum_{\substack{y\in\bbF_p^{n_j}\\\sum_{t=1}^{n_j}y_t=0\bmod p}}\prod_{t=1}^{n_j}\tf_j(y_t) &= \cdots = \begin{cases}\sqrt{p^{n_j}} &\text{if $n_j$ even,}\\\sqrt{p^{n_j}}\frac{2r-p}{p} &\text{if $n_j$ odd.}\end{cases}
    \end{align}
    The derivation assumed $n_j\geq2$, but the formula also holds for $n_j=0,1$.
    For $n_j=0$, the product contains no elements, and the sum is trivial, producing $1$ as answer.
    For $n_j=1$, there is only one summand $y=0$, and only one multiplicative factor $\tf_j(0)$, which we recall evaluates to $(2r-p)/\sqrt{p}$.
    We can now compute $F(n_1,\ldots,n_m)$ by multiplying the factors corresponding to each $n_j$, in order to arrive at
    \begin{align}
        F(n_1,\ldots,n_m) &= \prod_{j=1}^m\sum_{\substack{y\in\bbF_p^{n_j}\\\sum_{t=1}^{n_j}y_t=0\bmod p}} \prod_{t=1}^{n_j}\tf_j(y_t)
        = \prod_{j=1}^m \sqrt{p^{n_j}} \left(\frac{2r-p}{p}\right)^{\mathds1[n_j\in2\bbN+1]} \\
        &= \sqrt{p^{\sum_{j=1}^m n_j}}\left(\frac{2r-p}{p}\right)^{\sum_{j=1}^m\mathds1[n_j\in2\bbN+1]} = \sqrt{p^k}\left(\frac{2r-p}{p}\right)^{\lvert\{j\in[m]\,|\,n_j\in2\bbN+1\}\rvert}.
    \end{align}
    Here, we have introduced the notation $2\bbN+1$ to refer to the set of odd integers and $\mathds{1}$ to refer to the indicator function.
    We conclude that the function $F$ depends on $n_1,\ldots,n_m$ only via how many of the numbers are indeed odd.
    We finally relate the sum $\sum_{i\in[m]^k}F(i)$ to a sum of $F(n_1,\ldots,n_m)$ over all choices of $n_1+\cdots+n_m=k$,
    \begin{align}
        \sum_{i\in[m]^k}F(i) = \sum_{n_1+\cdots +n_m=k}\binom{k}{n_1,\ldots,n_m}F(n_1,\ldots,n_m).
    \end{align}
    This identity featuring the multinomial coefficient follows from standard combinatorial arguments.
    With these, we can give an interpretable expression
    \begin{align}
        \bbE_{x\sim\operatorname{Unif}}[f^k(x)] &=\cdots = \frac{1}{\sqrt{p^k}} \sum_{n_1+\cdots n_m=k}\binom{k}{n_1,\ldots,n_m}F(n_1,\ldots,n_m) \\
        &= \frac{1}{\sqrt{p^k}} \sum_{n_1+\cdots +n_m=k}\binom{k}{n_1,\ldots,n_m}\sqrt{p^k}\left(\frac{2r-p}{p}\right)^{\lvert\{j\in[m]\,|\,n_j\in2\bbN+1\}\rvert} \\
        &= \sum_{n_1+\cdots +n_m=k}\binom{k}{n_1,\ldots,n_m}\left(\frac{2r-p}{p}\right)^{\lvert\{j\in[m]\,|\,n_j\in2\bbN+1\}\rvert}
    \end{align}
    for the moments of the objective function under the uniform distribution. 
    Indeed, the $k^\text{th}$ moment of the objective function is given by the number of odd summands among all possible ways of choosing $k$ non-negative integers adding up to $k$.
    
    We now rewrite this sum in terms of an expectation value of a binomial-distributed random variable.
    We first replace the indicator function for $n_j$ being odd with a simpler expression, for any individual $j\in\{1,\ldots,m\}$, to get
    \begin{align}
        \left(\frac{2r-p}{p}\right)^{\mathds{1}[n_j\in2\bbN+1]} &= \frac{1+\frac{2r-p}{p}}{2} + \frac{1-\frac{2r-p}{p}}{2}(-1)^{n_j} = \frac{r}{p}+\frac{p-r}{p}(-1)^{n_j} = \frac{r}{p}+\left(1-\frac{r}{p}\right)(-1)^{n_j}.
    \end{align}
    We next use this expression to consider the case where the exponent contains a sum over indicator functions:
    \begin{align}
        \left(\frac{2r-p}{p}\right)^{\sum_{j=1}^m\mathds{1}[n_j\in2\bbN+1]} &= \prod_{j=1}^m \left(\frac{2r-p}{p}\right)^{\mathds{1}[n_j\in2\bbN+1]} = \prod_{j=1}^m\left(\frac{r}{p}+\left(1-\frac{r}{p}\right)(-1)^{n_j}\right) \\
        &= \sum_{S\subseteq[m]}\left(\frac{r}{p}\right)^{m-\lvert S\rvert}\prod_{j\in S}\left(\left(1-\frac{r}{p}\right)(-1)^{n_j}\right) \\
        &= \sum_{S\subseteq[m]}\left(\frac{r}{p}\right)^{m-\lvert S\rvert}\left(1-\frac{r}{p}\right)^{\lvert S\rvert}\left(-1\right)^{\sum_{j\in S}n_j}.
    \end{align}
    We can now combine this in the complete sum with multinomial weights.
    In particular, we re-organize the terms to replace the multinomial coefficients by the power of a sum
    \begin{align}
        \bbE_{x\sim\operatorname{Unif}}[f^k(x)] &= \cdots = \sum_{n_1+\cdots+n_m=k}\binom{k}{n_1,\ldots,n_m} \left(\frac{2r-p}{p}\right)^{\sum_{j=1}^m\mathds{1}[n_j\in2\bbN+1]} \\
        &= \sum_{n_1+\cdots+n_m=k}\binom{k}{n_1,\ldots,n_m}\sum_{S\subseteq[m]}\left(\frac{r}{p}\right)^{m-\lvert S\rvert}\left(1-\frac{r}{p}\right)^{\lvert S\rvert} (-1)^{\sum_{j\in S}n_j} \\
        &= \sum_{S\subseteq[m]}\left(\frac{r}{p}\right)^{m-\lvert S\rvert}\left(1-\frac{r}{p}\right)^{\lvert S\rvert}\sum_{n_1+\cdots+n_m=k}\binom{k}{n_1,\ldots,n_m}(-1)^{\sum_{j\in S}n_j} \\
        &= \sum_{S\subseteq[m]}\left(\frac{r}{p}\right)^{m-\lvert S\rvert}\left(1-\frac{r}{p}\right)^{\lvert S\rvert}\left(\sum_{j=1}^m(-1)^{\mathds{1}[j\in S]}\right)^k \\
        &= \sum_{S\subseteq[m]}\left(\frac{r}{p}\right)^{m-\lvert S\rvert}\left(1-\frac{r}{p}\right)^{\lvert S\rvert}\left(-1\cdot\lvert S\rvert+1\cdot(m-\lvert S\rvert)\right)^k \\
        &= \sum_{S\subseteq[m]}\left(\frac{r}{p}\right)^{m-\lvert S\rvert}\left(1-\frac{r}{p}\right)^{\lvert S\rvert}\left(m-2\lvert S\rvert\right)^k.
    \end{align}
    Here, we have used the multinomial theorem for the sum $(x_1+\ldots+x_m)^k$ with $x_j=-1$ if $n_j\in S$ and $x_j=1$ otherwise.
    We note that the final expression depends only on the size of the auxiliary subsets $S$, we thus introduce the variable $j=m-\lvert S\rvert\in\{0,\ldots,m\}$ and account for all possible different subsets of the same size, exploiting the symmetry $\binom{m}{j}=\binom{m}{m-j}$:
    \begin{align}
        \bbE_{x\sim\operatorname{Unif}}[f^k(x)] &= \cdots = \sum_{j=0}^m \binom{m}{j}\left(\frac{r}{p}\right)^{j}\left(1-\frac{r}{p}\right)^{m-j}\left(2j-m\right)^k.
    \end{align}
    To complete the proof, we introduce a random variable $J\in\{0,\ldots,m\}$ which follows a binomial distribution 
    \begin{align}
        \bbP_{\operatorname{Binom}}[J=j]=\binom{m}{j}\left(\frac{r}{p}\right)^{j}\left(1-\frac{r}{p}\right)^{m-j}
    \end{align}
    with parameter $r/p$.
    Then, we can interpret the sum as the expectation value of the function $(2J-m)^k$,
    \begin{align}
        \bbE_{J\sim\operatorname{Binom}\left(m, r/p\right)}\left[(2J-m)^k\right] &\coloneqq \sum_{j=0}^m\bbP_{\operatorname{Binom}}[J=j]\left(2j-m\right)^k = \sum_{j=0}^m\binom{m}{j}\left(\frac{r}{p}\right)^{j}\left(1-\frac{r}{p}\right)^{m-j}\left(2j-m\right)^k
    \end{align}
    with respect to $J$.
    All together, we reach the claimed statement
    \begin{align}
        \bbE_{x\sim\operatorname{Unif}}[f^k(x)] &= \sum_{j=0}^m\binom{m}{j}\left(\frac{r}{p}\right)^{j}\left(1-\frac{r}{p}\right)^{m-j}\left(2j-m\right)^k = \bbE_{J\sim\operatorname{Binom}\left(m,\frac{r}{p}\right)}\left[(2J-m)^k\right].
    \end{align}
\end{proof}

    \begin{corollary}[Consistency]
        Lemma~\ref{l:distance_uniform_distribution_moments_max-XORSAT} is a special case of Theorem~\ref{thm:distance_uniform_distribution_moments_maxLINSAT} with $p=2, r=1$.
    \end{corollary}
    \begin{proof}
        The proof is direct, since for $r=1,p=2$, we have $\operatorname{Binom}(m,r/p)=\operatorname{Binom}(m,1/2)$.
    \end{proof}
    
    \begin{remark}[Independence of $F$:]
        Note that the definition of the objective function $f$ depends both on (1) the matrix of coefficients $B$, and (2) the choice of satisfiability sets $F$.
        The first $d$ moments of the uniform distribution do \emph{not} depend on the choice of $F$, but rather only on the size of the sets $r$ and the matrix $B$, and then only via 
        the \emph{distance} $d$ of the error-correcting code with parity-check matrix $B^\intercal$.
        Because of this, it follows that the expectation value of the objective function $f(x)$ under the DQI distribution is itself independent of the choice of $F$ of size $r$, which we present in Theorem~\ref{thm:DQI_performance_simplified}.
    \end{remark}

    \begin{customthm}{\ref{thm:DQI_performance_simplified}}[Moments of the DQI distribution]
        Suppose we use DQI with a polynomial of degree $\ell$ fulfilling $2\ell+1\leq d$, where $d$ is the distance of the corresponding error-correction code.
        Then, for any $k < d - 2\ell$, the $k^\text{th}$ moment  of the objective function $f(x)$ under the DQI distribution fulfills
    \begin{align}
            \bbE_{x\sim \PDQI}\left[f(x)^k\right] &= p^n  \bbE_{J\sim\operatorname{Binom}\left(m,\frac{r}{p}\right)}\left[\sum_{k'=0}^{2\ell}\beta_{k'}(2J-m)^{k+k'}\right].
        \end{align}
    \end{customthm}
    \begin{proof}
        We prove the statement directly, by combining Lemma~\ref{l:moments_dqi_unif} and Theorem~\ref{thm:distance_uniform_distribution_moments_maxLINSAT}, to get
        \begin{align}
            \bbE_{x\sim \PDQI}\left[f(x)^k\right] &= p^n\sum_{k'=0}^{2\ell} \beta_{k'}\bbE_{x\sim\operatorname{Unif}}\left[f^{k+k'}(x)\right]
            = p^n \sum_{k'=0}^{2\ell} \beta_{k'} \bbE_{J\sim\operatorname{Binom}\left(m,\frac{r}{p}\right)}\left[(2J-m)^{k+k'}\right] \\
            &= p^n  \bbE_{J\sim\operatorname{Binom}\left(m,\frac{r}{p}\right)}\left[\sum_{k'=0}^{2\ell}\beta_{k'}(2J-m)^{k+k'}\right].
        \end{align}
    \end{proof}

\section{Efficient evaluation of DQI probabilities
\label{a:efficient_probabilities}}
We describe how to efficiently compute the DQI probabilities $P^2(f(x))$ for a given $x$ during MCMC, which is critical for running the algorithm at large scales.

\paragraph{Lookup table.} The objective $f(x)$ takes only $m+1$ distinct integer values since it is the difference of satisfied and violated constraints.
We therefore precompute $P^2(y)$ for all $y\in\{-m,-m+2,\ldots,m\}$ and store the results in a lookup table of size $O(m)$.
Evaluating $P^2(f(x))$ then costs $O(1)$.
The objective $f(x)$ depends on the number of satisfied constraints in a max-LINSAT instance $(B, \{F_i\})$. 
Its efficient evaluation depends on the structure of $B$, so we treat the two problem families separately.

\paragraph{Max-XORSAT (sparse $B$, $p=2$).}
Here $B$ is $s$-row-sparse and $t$-column-sparse, i.e., each row has at most $s$ nonzeros and each column at most $t$.
We cache the vector $\mathtt{bx}=Bx \mod 2\in\mathbb{F}_2^m$ across steps.
Given $\mathtt{bx}$, computing $n_{\text{sat}} = |\{i : \mathtt{bx}_i = v_i\}|$ and hence $f(x) = 2n_{\text{sat}} - m$ costs $O(m)$.  
Flipping a single variable $x_j$ affects only the entries of $\mathtt{bx}$ indexed by the nonzero rows of column $j$ of $B$, so the update costs $O(t)$.  
To evaluate all $2^\kappa$ assignments for a block of $k$ variables, we enumerate them in Gray-code order so consecutive candidates differ in exactly one variable. 
The total work per step is therefore $O(2^\kappa(t + m))$: $O(t)$ to update $\mathtt{bx}$ per Gray-code transition and $O(m)$ to recompute $n_{\text{sat}}$ per candidate.
Initialization (forming $\mathtt{bx}$ via a sparse matrix--vector product) costs $O(ms)$, giving a total cost of $O(ms + 2^\kappa m T)$ after $T$ steps.  
For fixed $\kappa$ and constant sparsity parameters $s,t$, with $t \ll m$, the per-step cost is $O(m)$.  
In practice, the inner loop over Gray-code candidates is JIT-compiled via \textsc{Numba}~\cite{lam2015numba}, eliminating Python overhead and making the $O(m)$ scan per candidate fast enough for the instance sizes in Section~\ref{s:results}.

\paragraph{OPI (dense Vandermonde $B$, large prime $p$).}
Here $B$ is a dense $m \times n$ Vandermonde matrix with entries $B_{ij} = \gamma^{i(j-1)}\bmod p$.  
The score $n_{\text{sat}}$ is computed evaluating $Q(\gamma^i) = \sum_{j=1}^{n} x_j \gamma^{i(j-1)} \bmod p$ for each $i$ and checking membership in $F_i$ via a precomputed boolean mask, which requires a full $O(mn)$ matrix--vector product.  
Unlike max-XORSAT, no incremental $Bx$ cache is maintained between candidates: $f$ is recomputed for each of the $p^\kappa$ candidate assignments in the block, giving a per-step cost of $O(p^\kappa mn)$.
The one-time precomputation of the $m \times n$ power table $\gamma^{i(j-1)} \bmod p$ and the $m \times p$ membership mask costs $O(mn)$ and $O(mp)$ respectively. The internal loop is again JIT-compiled via \textsc{Numba}, and split into
blocks that are run in parallel using \textsc{Numba}'s \texttt{prange}.

\section{Extended discussion on the complexity of DQI}\label{a:extended_complexity}

\subsection{Obstacle against proving hardness of sampling}\label{aa:counting}

One possible direction to tackle the complexity of sampling would be to reduce the sampling problem into a counting problem.
Counting problems have been used to argue about the hardness of other sampling problems based on quantum circuits, like random circuit sampling and other quantum random sampling schemes \cite{SupremacyReview}.
We argue that the mechanisms for the potential hardness of sampling are different there.

For DQI, one natural counting problem would be to evaluate the 
function $M(y)\coloneqq\lvert\{x\in\bbF_p^n\,|\,f(x)=y\}\rvert$.
In general, computing $M$ exactly is $\#\mathsf{P}$-hard~\cite{valiant1979complexity}, so we do not expect efficient algorithms (quantum or classical) to succeed.
Recall that DQI is a sampling algorithm whose outcome probabilities can be computed efficiently classically.
Having access to samples from the DQI distribution is not enough in general to compute $M$ exactly.
Rather, having access to samples from the DQI distribution allows us to estimate $M$ to a certain additive precision.
To set up a concrete comparison, we consider two classical players, Alice and Bob, and we assume that Alice has black-box access to polynomially-many exact samples from the DQI distribution.
The question is then: \emph{can Alice efficiently estimate $M$ to a certain precision which Bob cannot efficiently replicate?}

We argue heuristically why Bob may be able to replicate Alice's capabilities in general, up to polynomial overheads.
Alice can obtain samples efficiently from the DQI distribution, from which she can estimate certain moments $\bbE_{x\sim\PDQI}[f^k(x)]$.
In absence of further resources, the precision to which Alice can estimate the counting function $M$ will depend on the moments she can estimate: with higher $k$ yielding the better precision.
The same fact applies to Bob, with the difference that he does not have access to the samples, and thus must consider alternative ways to estimate those moments.
From Lemma~\ref{l:moments_dqi_unif}, we know how the moments Alice obtains relate to the moments of the uniform distribution, and further Theorem~\ref{thm:distance_uniform_distribution_moments_maxLINSAT} relates them to those of a binomial distribution.
In this sense, we expect Bob to be able to reach a similar level of precision as Alice, based on both their abilities to estimate certain moments of the objective function.

Strictly speaking, our results only work for $k$ satisfying $k+2\ell\leq d$, where $\ell$ is the degree of the polynomial used by DQI, and $d$ is the distance of the error-correcting code.
A priori there is no reason for Alice to set $\ell \ll d/2$, and hence Alice may be able to use the samples to estimate higher-moments than those available from our results.
Still, it may be the case that the available efficient decoders for broad classes of problems cannot correct large numbers of errors.
In any case, we expect that an approximate version of Theorem~\ref{thm:distance_uniform_distribution_moments_maxLINSAT} still holds if one goes slightly beyond the code distance (as is the case for the performance guarantees of DQI~\cite{Jordan2025optimization}).

Finally, we highlight one heuristic argument supporting that higher moments translate to better precision.
Theorem~\ref{thm:distance_uniform_distribution_moments_maxLINSAT} relates the distance of the error-correcting code to the distribution of scores.
If we drew an instance of max-LINSAT uniformly at random, we would expect the distribution of scores to be exactly binomially-distributed.
Estimating $M$ is easy in this case, since $M(y)$ would closely follow $\operatorname{Binom}(m,r/p)$.
The families of problems we discuss in the main text, max-XORSAT and OPI, are not drawn uniformly at random, and hence their score distribution is not exactly binomial-distributed.
Still, we can use Theorem~\ref{thm:distance_uniform_distribution_moments_maxLINSAT} to argue that the higher the 
distance of the dual code for a given instance of max-LINSAT, the closer its distribution of scores becomes to the uniform distribution (in the sense that more of their moments agree).
Then, in a way, the \emph{counting} problem that we could solve with DQI becomes \emph{easier} in the regimes where DQI ought to give an advantage.

With this, we only point out 
that there may be an obstruction in proving that sampling from the DQI distribution is hard for efficient classical algorithms via a reduction to counting. Our arguments do not affect our ability to efficiently implement DQI on a quantum computer, assuming an efficient decoder is known.

\subsection{Obstacles against the standard reduction from search to decision}\label{aa:decision}

Let us consider an abstract optimization problem specified by an objective function $f:\bbF_2^n\to\bbZ$.
One natural way to define a \emph{decision} problem would be, given a threshold $\alpha\in\bbZ$, decide between:
\begin{itemize}
    \item \textbf{YES}: if there exists $x\in\bbF_2^n$ such that $f(x)\geq\alpha$.
    \item \textbf{NO}: for all $x\in\bbF_2^n$, it holds $f(x) < \alpha$.
\end{itemize}
One natural way to define a \emph{search} problem would be, given a threshold $\alpha\in\bbZ$, and under the promise that the following exists, find $x\in\bbF_2^n$ such that $f(x)\geq\alpha$.

Let us introduce the set of \emph{guesses} $G=\{0,1,\bot\}^n$.
Indeed, the entries of a guess $g=(g_j)_{j=1}^n\in G$ may be either $g_j\in\bbF_2$ or $g_j=\bot$, which we call \emph{indefinite}.
Given a guess $g\in G$, we say it is $N$-definite, with $N\in\{0,\ldots,n\}$ if $N$ of its entries are non-$\bot$.
Given $f$ and an $N$-definite guess $g\in G$, we introduce the \emph{reduced problem} $f_g:\bbF_2^{n-N}\to\bbZ$, which corresponds to the original problem up to fixing $x$ to agree with $g$ in the $N$ entries that are not $\bot$.

Given $f$ and $\alpha$, suppose we have access to an oracle $O_\alpha:G\to\{\textbf{YES},\textbf{NO}\}$ which can solve the decision problem and all its reduced versions.
Then, we can solve the search problem with threshold $\alpha$ by querying the oracle exactly $n$ times, by iteratively building up a complete guess:
\begin{enumerate}
    \item Initialize $g\gets(\bot)^n$, which is $0$-definite.
    \item Repeat $n$ times, on the $t^\text{th}$ step:
    \begin{enumerate}
        \item Fix the $t^\text{th}$ component of the guess $g_{t}\gets0$, which is now $t$-definite.
        \item Query the oracle. If $O_\alpha(g)=\textbf{YES}$, move to the next iteration; else correct the $t^\text{th}$ component of the guess $g_{t}\gets 1$.
    \end{enumerate}
    \item Output $x=g$.
\end{enumerate}
Under the promise that a 
solution exists, $O_\alpha\left((\bot)^n\right)=\textbf{YES}$, we reach a solution $x$ after $n$ queries.
The iterative procedure ensures that there exists a solution which agrees with the definite entries of the growing guess at all times, hence no back-tracking is required.
This procedure can be straightforwardly generalized on fields $\bbF_p$, where every iteration requires testing up to $p-1$ values per coordinate.

In the main text, we have not emphasized this fact, but one could interpret the performance guarantees of Ref.~\cite{Jordan2025optimization} as an oracle.
Let $f:\bbF_p^n\to\{-m,-m+2,\ldots,m\}$ be the objective function of max-LINSAT, with 
prime $p$, $m$ constraints, $n$ variables, and $r$-sized sets.
Suppose we have an efficient decoder for errors of weight up to $\ell$.
Then, the expected number of satisfied constraints $\sDQI$ using DQI fulfills
\begin{align}\label{eq:perf_guarantee}
    \frac{\sDQI}{m} &=\left(\sqrt{\frac{\ell}{m}\left(1-\frac{r}{p}\right)}+\sqrt{\left(1-\frac{\ell}{m}\right)\frac{r}{p}}\right)^2.
\end{align}
This fact alone tells us that there exist $x\in\bbF_p^n$ such that $f(x)\geq 2\sDQI-m$.
One may then wonder whether this knowledge suffices to instantiate the oracle introduced above with a carefully chosen $\alpha=\alpha(\ell,m,r,p)$, or a probabilistic version thereof.
We argue that such an approach would not work in general, due to the interplay between the number of decodable errors and the reduced problem for max-LINSAT.
Given an $N$-definite guess $g\in G$, it may not be easy to assess whether there exist $x\in\bbF_p^n$ such that $f_g(x)\geq 2\sDQI-m$.
Trying to answer this question directly would involve using Eq.~\eqref{eq:perf_guarantee} for the reduced problem $f_g$, which would require us to know the corresponding number of decodable errors.
We identify at least two main obstacles:
\begin{enumerate}
    \item The reduced problem arising from $N$-definite guesses for large $N$ involves only a few variables $n-N$.
        In the reduction from optimization to decoding, we observe that the number of variables in the optimization problem corresponds to the number of parity checks in the error-correcting code.
        Knowing that an error-correcting code with few checks cannot have large distance, it follows that $\sDQI/m$ for the reduced problem converges to $r/p$, and hence we cannot use Eq.~\eqref{eq:perf_guarantee} to instantiate the oracle with the threshold $\alpha$ corresponding to the original problem.
        This obstacle prevents us from realizing at least the last few steps of the procedure above, since the decision oracle \enquote{breaks down} for small instance sizes.
    \item The performance guarantee assumes not only that $\ell$ is known, but also that an efficient decoder is provided.
        From decoding being $\mathsf{NP}$-hard, it follows that certain structure is required for DQI to perform well.
        For us to be able to use Eq.~\eqref{eq:perf_guarantee} throughout the procedure we would then need the added requirement that the structure is preserved under reducing the problem for increasingly-definite guesses.
\end{enumerate}

To be precise: the above obstacles only prevent a direct reduction from search to decision.
These obstacles are fully unrelated to DQI's ability to solve the search (and sampling) problems based on certain families of max-LINSAT.
This way, it seems that making statements about the complexity of DQI requires studying the search problem, and not only the decision one.

\section{Further details numerical results}\label{a:further_details}

    In this section, we provide extended data from our numerical experiments, complementing the summary figures in Section~\ref{s:results} in the main text.
    Each of the following figures is a direct expansion of some figure in the main text, namely: Fig.~\ref{fig:ell_dqi_max-XORSAT_scaling} expands on Fig.~\ref{fig:ell_dqi_max-XORSAT}a, Fig.~\ref{fig:tau_search_scaling} expands on Fig.~\ref{fig:tau_search}a and c, Figs.~\ref{fig:tau_sampling_scaling_max-XORSAT} and~\ref{fig:tau_sampling_scaling_OPI} expand on Fig.~\ref{fig:tau_sampling}a and c, respectively.
    We also collect the numerical fits on all plots in Table~\ref{tab:fits}.
    
    Let us consider a dataset $X, Y$, with $X=\{x_i\}_{i=1}^N$ and $Y=\{y_i\}_{i=1}^N$, as well as the values produced by a candidate fit $Z=\{z_i\}_{i=1}^N$.
    Suppose the fit is a function $h$ such that $z_i=h(x_i)$, and $h$ has been chosen with the goal of approximating $Y$: $z_i\approx y_i$.
    Let \begin{equation}\bar y\coloneqq\frac{1}{N}\sum_{i=1}^Ny_i.
    \end{equation}
    Then, we consider the usual 
    \begin{align}
        R^2(Y,Z) &= 1 - \frac{\sum_{i=1}^N(y_i-z_i)^2}{\sum_{i=1}^N(y_i-\bar y)^2}
    \end{align}
    quantity as our main figure of merit for \emph{goodness-of-fit}.
    Since we are interested in qualitative scaling behavior, we consider the multiplicative variant 
    \begin{align}
        R^2_\times(Y,Z) &= 1 - \frac{\sum_{i=1}^N(\log(y_i)-\log(z_i))^2}{\sum_{i=1}^N(\log(y_i)-\bar y_{\log})^2}
    \end{align}
    as well, for which we define $\bar y_{\log} = \frac{1}{N}\sum_{i=1}^N\log(y_i)$ (which is well defined for $y_i>0$).
    Due to the properties of the logarithm, if the fit $Z$ only disagrees with $Y$ by small multiplicative factors, the score $R_\times^2$ can still be reasonably high.
    Furthermore, the procedure we use for fitting both power-laws and exponentials explicitly aims to minimize the square difference of logarithms, thus further justifying this figure of merit.
    Since $R_\times^2$ is more forgiving in terms of relative errors, we note that both power-law and exponential fits can be consistent with the same data, as shown in Fig.~\ref{fig:comparison_thresholds}a.

\begin{table}[t]
    \centering
    \caption{
        \textbf{Numerical data fits.}
        Each row clearly specifies each of the fits appearing in all the figures.
        Refer to Sec.~\ref{a:further_details} for the definition of $R^2$ and $R_\times^2$ as quality metrics.
    }
    \begin{tabular}{cccc}\hline\hline
        Experiment & Reference & Fit & Quality \\\hline\hline
        DQI polynomial degree & Fig.~\ref{fig:ell_dqi_max-XORSAT}c & $\ell_{\operatorname{dqi}}=0.20n -72$ & $R^2=0.999$ \\
        Search max-XORSAT & Fig.~\ref{fig:tau_search}b & $\tau_{\avg} = 1.88 \times 10^{-8}\, n^{4.90}$  & $R^2_\times=0.978$ \\
        Search OPI & Fig.~\ref{fig:tau_search}d & $\tau_{\min} = 23\, \left(1.096\right)^{n_p}$  & $R^2_\times=0.963$ \\
        Sampling max-XORSAT \emph{restart} & Fig.~\ref{fig:tau_sampling}b & $\tau_{10} = 1.38 \times 10^{-6}\, n^{4.51}$  & $R^2_\times=0.970$ \\
        Sampling max-XORSAT \emph{keep-going} & Fig.~\ref{fig:tau_sampling}b & $\tau_{10} = 2.67 \times 10^{-6}\, n^{4.02}$  & $R^2_\times=0.942$ \\
        Sampling OPI \emph{restart} & Fig.~\ref{fig:tau_sampling}d & $\tau_{10} = 41\, \left(1.104\right)^{n_p}$  & $R^2_\times=0.962$ \\
        Sampling OPI \emph{keep-going} & Fig.~\ref{fig:tau_sampling}d & $\tau_{10} = 266\, \left(1.103\right)^{n_p}$  & $R^2_\times=0.961$ \\\hline\hline
    \end{tabular}
    \label{tab:fits}
\end{table}

\begin{figure}[t]
    \centering
    \includegraphics{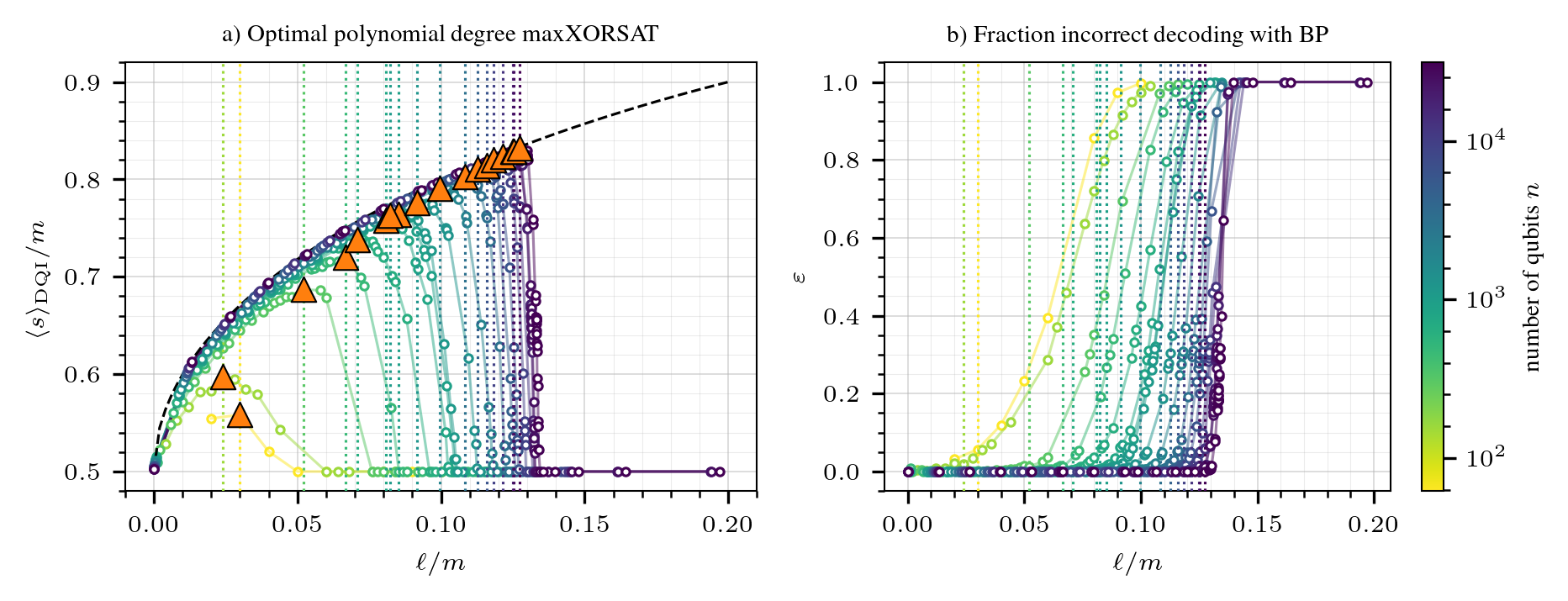}
    \caption{
        \textbf{Scaling behavior of DQI for max-XORSAT -- extended.}
        a) Expected performance of DQI vs normalized polynomial degree for all instance sizes, see Eq.~\eqref{eq:dqi_score_eps}.
        Each individual line is the equivalent to the orange points in Fig.~\ref{fig:ell_dqi_max-XORSAT}a.
        The dashed line corresponds to the semi-circle law from Eq.~\eqref{eq:dqi_score}.
        b) Fraction of incorrect decoding vs normalized polynomial degree for all instance sizes.
        Each individual line is the equivalent to the green points in Fig.~\ref{fig:ell_dqi_max-XORSAT}a.
        The color encodes the number of qubits in both panels.
        The vertical dotted lines highlight the fraction $\ell/m$ at which the maximum is achieved for every size.
    }
    \label{fig:ell_dqi_max-XORSAT_scaling}
\end{figure}

\begin{figure}[t]
    \centering
    \includegraphics{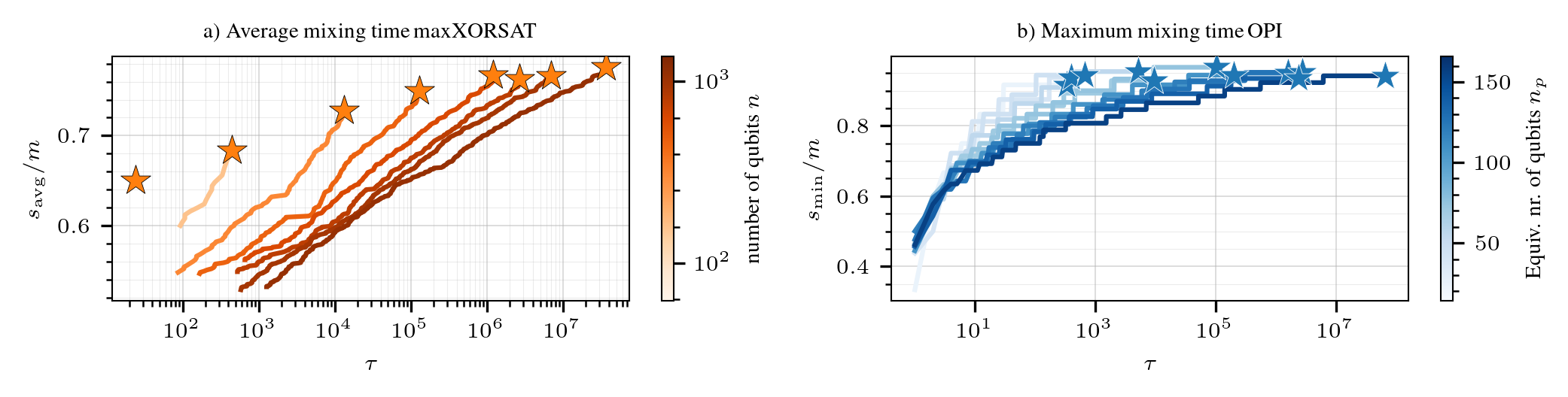}
    \caption{
        \textbf{Scaling behavior of Markov-chain mixing times -- extended.}
        a) Average MCMC trajectory for all instance sizes for max-XORSAT, with $K=100$ independent runs for each instance size.
        Each individual line is equivalent to that of Fig.~\ref{fig:tau_search}a, for different sizes.
        The orange stars correspond to those in Fig.~\ref{fig:tau_search}b.
        b) Average MCMC trajectory for all instance sizes for OPI, with $K=100$ independent runs for each instance size.
        Each individual line is equivalent to that of Fig.~\ref{fig:tau_search}c, for different sizes.
        The blue stars correspond to those in Fig.~\ref{fig:tau_search}d.
        We identify the trajectory corresponding to the smallest equivalent number of qubits as an outlier.
        }
    \label{fig:tau_search_scaling}
\end{figure}

\begin{figure}[t]
    \centering
    \includegraphics{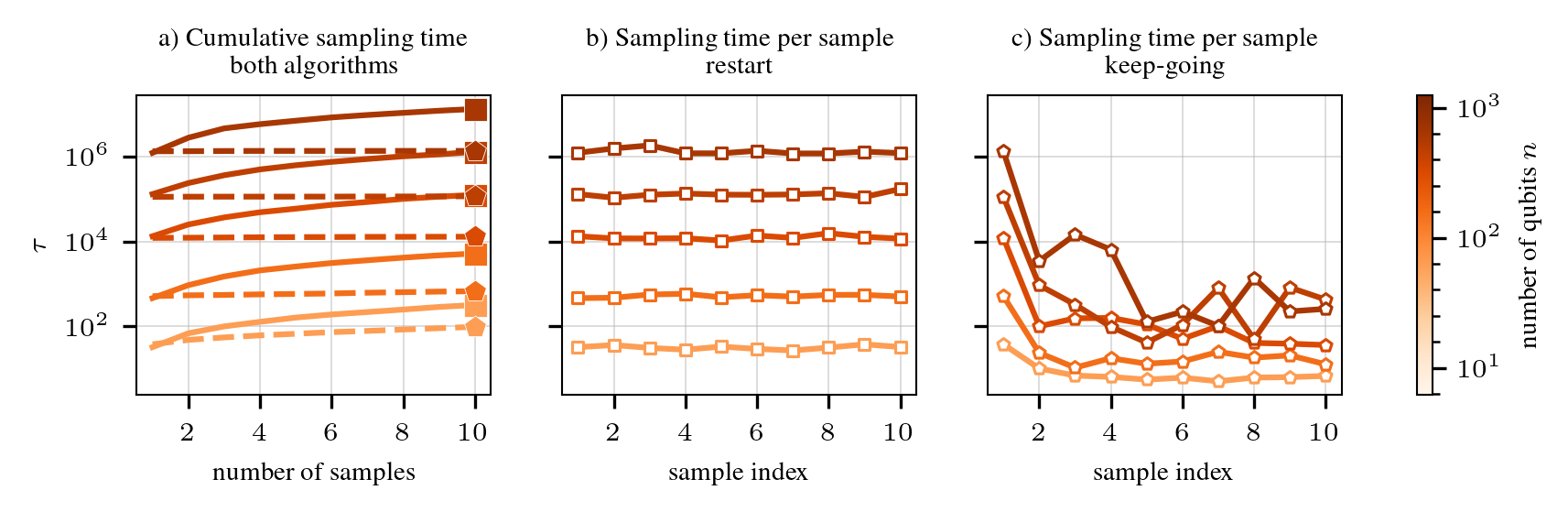}
    \caption{
        \textbf{Scaling behavior of Markov-chain sampling times -- max-XORSAT.}
        a) Average MCMC sampling time for all instance sizes, with $K=100$ independent runs for each instance size.
        The dashed line corresponds to the \emph{keep-going} algorithm, the solid line corresponds to the \emph{restart} algorithm.
        Each individual pair of lines is equivalent to those in Fig.~\ref{fig:tau_sampling}a.
        The squares and pentagons correspond to those in Fig.~\ref{fig:tau_sampling}b.
        b) Average cost \emph{per sample} for the \emph{restart} algorithm, for all instance sizes.
        c) Average cost \emph{per sample} for the \emph{keep-going} algorithm, for all instance sizes.
        The color encodes the number of qubits in all three panels.
        }
    \label{fig:tau_sampling_scaling_max-XORSAT}
\end{figure}

\begin{figure}[t]
    \centering
    \includegraphics{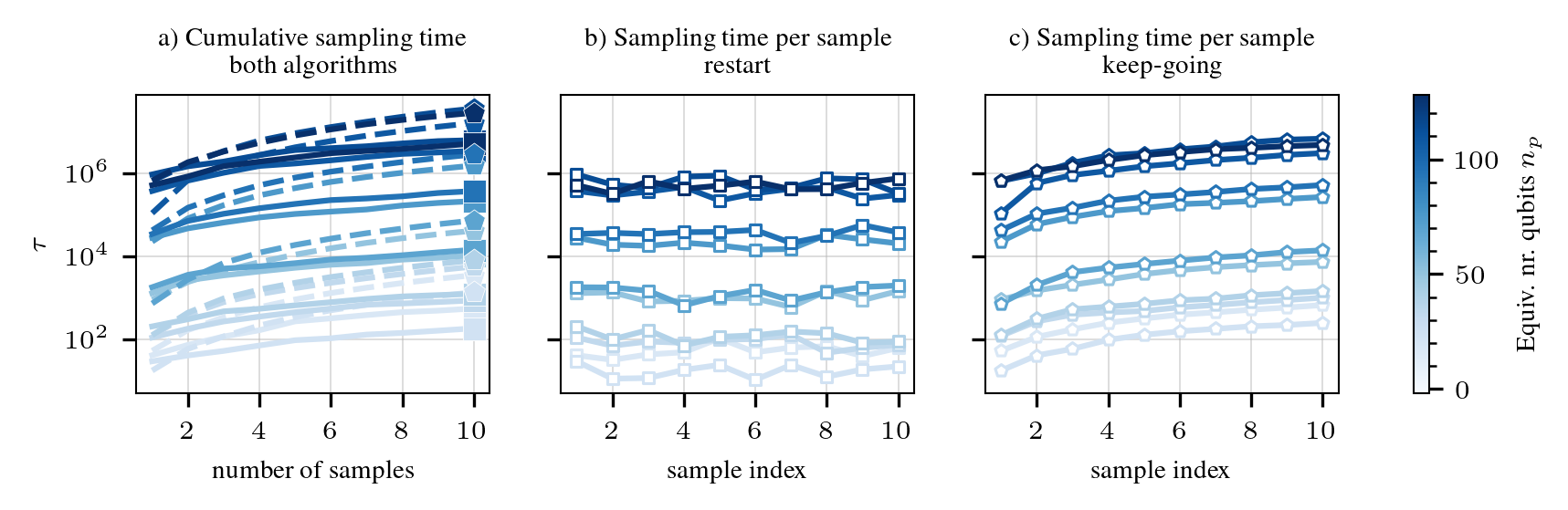}
    \caption{
        \textbf{Scaling behavior of Markov-chain sampling times -- OPI.}
        a) Average MCMC sampling time for all instance sizes, with $K=10$ independent runs for each instance size.
        The dashed line corresponds to the \emph{keep-going} algorithm, the solid line corresponds to the \emph{restart} algorithm.
        Each individual pair of lines is equivalent to those in Fig.~\ref{fig:tau_sampling}c.
        The squares and pentagons correspond to those in Fig.~\ref{fig:tau_sampling}d.
        b) Average cost \emph{per sample} for the \emph{restart} algorithm, for all instance sizes.
        c) Average cost \emph{per sample} for the \emph{keep-going} algorithm, for all instance sizes.
        The color encodes the number of qubits in all three panels.
        }
    \label{fig:tau_sampling_scaling_OPI}
\end{figure}

\end{document}